\documentclass[11pt]{article}
\usepackage[T1]{fontenc}
\usepackage{geometry} 
\usepackage{graphicx}
\usepackage{amsmath,amssymb,amsthm,amsfonts,dsfont,mathtools}
\usepackage{relsize}
\usepackage[margin=1cm]{caption}
\usepackage{wrapfig}
\usepackage{frame,color}
\usepackage{environ,subfig}
\usepackage{hyperref}
\usepackage{xspace}
\usepackage{framed}
\usepackage{comment}
\usepackage{changepage}
\usepackage{algorithm}
\usepackage[noend]{algpseudocode}
\usepackage{mathrsfs}
\usepackage{soul}

\algnewcommand{\Inputs}[1]{%
  \State \textbf{Inputs:}
  \Statex \hspace*{\algorithmicindent}\parbox[t]{.8\linewidth}{\raggedright #1}
}
\algnewcommand{\Initialize}[1]{%
  \State \textbf{Initialize:}
  \Statex \hspace*{\algorithmicindent}\parbox[t]{.8\linewidth}{\raggedright #1}
}
\algnewcommand{\TurnOne}[1]{%
  \State \textbf{Timestep 1:}
  \Statex \hspace*{\algorithmicindent}\parbox[t]{.8\linewidth}{\raggedright #1}
}

\makeatletter
 {\par\unskip\endMakeFramed}
\makeatother

\newtheorem{fact}{Fact}
\newtheorem{remark}[fact]{Remark}
\newtheorem{theorem}[fact]{Theorem}
\newtheorem{lemma}[fact]{Lemma}

\newtheorem{corollary}[fact]{Corollary}

\newcommand{\broadcast}{\mathsf{Broadcast}}
\newcommand{\cd}{\mathsf{CD}}
\newcommand{\nocd}{\mathsf{No}{\text -}\mathsf{CD}}
\newcommand{\sr}{\mathsf{SR}{\text -}\mathsf{communication}}
\newcommand{\leader}{\mathsf{Leader Election}}
\newcommand{\approxcounting}{\mathsf{Approximate Counting}}

\newcommand{\ignore}[1]{}

\newcommand{\Expect}{\operatorname{E}}

\newcommand{\Prob}{\operatorname{Pr}}
\newcommand{\bydef}{\stackrel{\rm def}{=}}

\newcommand{\fr}[2]{\mbox{$\frac{#1}{#2}$}}

\newcommand{\poly}{{\operatorname{poly}}}

\newcommand{\dist}{\operatorname{dist}}
\newcommand{\Color}{\operatorname{Color}}

\newcommand{\zero}[1]{\makebox[0mm][l]{$#1$}}

\newcommand{\rb}[2]{\raisebox{#1 mm}[0mm][0mm]{#2}}
\newcommand{\istrut}[2][0]{\rule[- #1 mm]{0mm}{#1 mm}\rule{0mm}{#2 mm}}

\newcommand{\ID}{\operatorname{ID}}
\newcommand{\CID}{\operatorname{CID}}
\newcommand{\LOCAL}{\mathsf{LOCAL}}

\NewEnviron{fcodebox}[1]%
{\fbox{%
\begin{minipage}{#1\linewidth}
\begin{codebox}
\BODY
\end{codebox}
\end{minipage}
}}

\title{The Energy Complexity of Broadcast\thanks{Supported by NSF grants CCF-1514383 and CCF-1637546.}}

\author{\istrut[0]{6}Yi-Jun Chang\\
University of Michigan
\and
Varsha Dani\\
University of New Mexico\\
\ \\
\and
Thomas P. Hayes\thanks{Supported by NSF CAREER award CCF-1150281.}\\
University of New Mexico\\
\and
Qizheng He\\
IIIS, Tsinghua University
\and
Wenzheng Li\\
IIIS, Tsinghua University
\and
Seth Pettie\\
University of Michigan
}

\begin{document}
\date{}
\maketitle
\thispagestyle{empty}
\setcounter{page}{0}

\begin{abstract}
\emph{Energy} is often the most constrained resource in networks of battery-powered devices,
and as devices become smaller, they spend a larger fraction of their energy on \emph{communication} (transceiver usage)
not computation.
As an imperfect proxy for true energy usage, we define \emph{energy complexity}
to be the number of time slots a device transmits/listens; idle time and computation are free.

In this paper we investigate the energy complexity of fundamental communication primitives such as $\broadcast$
in \emph{multi-hop} radio networks.  We consider models with collision detection ($\cd$) and without ($\nocd$),
as well as both randomized and deterministic algorithms.  Some take-away messages from this work include:

\begin{itemize}
\item The \emph{energy} complexity of $\broadcast$ in a multi-hop network
is intimately connected to the \emph{time} complexity of
$\leader$ in a single-hop (clique) network.
Many existing lower bounds on time complexity
immediately transfer to energy complexity.
For example, in the $\cd$ and $\nocd$ models, we
need $\Omega(\log n)$ and $\Omega(\log^2 n)$ energy, respectively.

\item The energy lower bounds above \emph{can} almost be achieved, given sufficient ($\Omega(n)$) time.
In the $\cd$ and $\nocd$ models we can solve $\broadcast$ using
$O(\frac{\log n\log\log n}{\log\log\log n})$ energy and $O(\log^3 n)$ energy, respectively.

\item The complexity measures of Energy and Time are in conflict, and it is an open problem whether both can be minimized simultaneously.
We give a tradeoff showing it is possible to be \emph{nearly} optimal in both measures simultaneously.
For any constant $\epsilon>0$, $\broadcast$ can be solved
in $O(D^{1+\epsilon}\log^{O(1/\epsilon)} n)$ time with
$O(\log^{O(1/\epsilon)} n)$ energy, where $D$ is the diameter of the network.
\end{itemize}
\end{abstract}
\newpage

\section{Introduction}\label{sect:intro}

In many networks of small wireless devices the scarcest resource is energy,
and the majority of energy is often spent on \emph{radio transceiver usage}---sending and receiving packets--- not
on computation~\cite{PolastreSC05,BarnesCMA10,Lee+13,SivalingamSA97}.
Rather than account for the energy profile of every mode of operation,
we simply assume that devices spend one unit of energy to send/listen and nothing for computation.
It is not uncommon to use transceiver usage as a proxy for total energy
(see, e.g.,~\cite{ChangKPWZ17,JurdzinskiKZ02,JurdzinskiKZ02b,JurdziskiKZ03,JurdzinskiKZ02c}).

\paragraph{The Model.}  The network is a connected undirected graph $G=(V,E)$ with devices associated with vertices.  Vertices know nothing about $G$, 
except for some general parameters such as the number of vertices $n = |V|$, the maximum degree $\Delta = \max_v \deg(v)$, and the diameter $D = \max_{u,v} \operatorname{dist}(u,v)$.\footnote{Each of $\Delta$ and $D$
can be upper bounded by $n$ if it is unknown.}
\emph{Time} is partitioned into discrete slots, and all vertices agree on time slot zero.
In each time slot, each device can choose to either
(i) send a message,
(ii) listen, or
(iii) remain idle,
where (i) and (ii) cost one unit of energy and (iii) is free.
If a device chooses to send a message or remain idle, it gets no feedback from the environment.
If a device chooses to listen and \emph{exactly one} neighbor sends a message $m$, it receives $m$.
The other cases depend on how the model handles collisions.
\begin{description}
\item[$\nocd$]  If zero or at least two neighbors transmit, a listener hears a signal $\lambda_S$, indicating \emph{silence}.

\item[$\cd$] If zero neighbors transmit, a listener hears $\lambda_S$; if at least two neighbors transmit, a listener hears $\lambda_N$, indicating \emph{noise}.

\item[$\LOCAL$] Every listener hears \emph{every} message transmitted by any neighbor.  There are no collisions.\footnote{Lower bounds in the $\LOCAL$ model are robust since
they capture the difficulty of synchronization, not on the subtleties of any particular collision-detection model.  This model bears the same name as Linial's $\LOCAL$ model~\cite{Linial92,Peleg00}
and is very similar to it.  In the traditional $\LOCAL$ model vertices do not have to choose between transmitting and listening, and there is no cost associated with communication.}
\end{description}

Finally, all the models come in randomized and deterministic variants.  In the deterministic setting, vertices are assigned distinct IDs in $\{1,\ldots,N\}$ and can use them to break symmetry.
Randomized algorithms can generate private random bits to break symmetry, e.g., they can assign themselves $O(\log n)$-bit IDs, which are distinct w.h.p.

\paragraph{Our Contribution.}
Previous work on energy complexity
has focussed on fundamental problems in single-hop (clique) networks like $\leader$ and
$\approxcounting$\footnote{(approximating `$n$' to within a constant factor)}~\cite{BenderKPY16,ChangKPWZ17,JurdzinskiKZ02,JurdzinskiKZ02b,JurdziskiKZ03,JurdzinskiKZ02c,kardas2013energy,NakanoO00},
where it is typical to assume that $n$ is unknown.
In this paper we consider fundamental problems in arbitrary \emph{multi-hop}
network topologies, primarily $\broadcast$.  At time zero there is a distinguished \emph{source} device $s\in V$
holding a message $m$.  By the end of the computation all vertices should know $m$.
We establish lower and upper bounds on $\broadcast$ in all collision-detection models,
both randomized and deterministic.
Table~\ref{table:broadcast} lists our results.  Some of the more interesting findings are as follows.
\begin{itemize}
\item \emph{Time} lower bounds on $\leader$ in single-hop networks extend to \emph{energy} lower bounds on $\broadcast$ in multihop networks.
As a consequence, we get energy lower bounds on $\broadcast$ of $\Omega(\log n)$ and $\Omega(\log\Delta\log n)$ in $\cd$ and $\nocd$, respectively.  
These lower bounds
reflect the difficulty of local \emph{contention resolution}, not on broadcasting \emph{per se}.  
We give a more robust energy lower bound of $\Omega(\log D)=\Omega(\log n)$
that reflects the difficulty of getting a message across a long path.  It applies to any collision-detection model, even $\LOCAL$.

\item Given sufficient time ($\Omega(n)$, regardless of the diameter $D$), these energy lower bounds can almost be achieved.
We give algorithms for $\cd$ and $\nocd$ using energy $O(\frac{\log n\log\log\Delta}{\log\log\log\Delta})$ and $O(\log\Delta\log^2 n)$, respectively.
Moreover, we show that on constant degree graphs, there is an energy-efficient preprocessing step that allows $\nocd$ to simulate $\LOCAL$.
This leads to an $O(n\log n)$ time, $O(\log n)$ energy $\broadcast$ algorithm when $\Delta=O(1)$.

\item Even with an infinite energy budget we need $\Omega(D)$ time.  We show that it is possible to achieve near optimality in both energy and time
simultaneously.  For any $\epsilon>0$, there is a $\broadcast$ algorithm taking $O(D^{1+\epsilon}\log^{O(1/\epsilon)} n)$ time and $O(\log^{O(1/\epsilon)} n)$ energy.

\item An interesting special case is when $G$ is a path, but the nodes do not know their position within this path, or the orientation of their neighbors.
In this setting, we are able to provide a provably optimal $\broadcast$ algorithm, taking $O(n)$ time and expected $O(\log n)$ energy.  
Neither time nor energy can be improved, even sacrificing the other.

\end{itemize}

\begin{table}
\centering
\begin{tabular}{|l|l|l|l|}
\multicolumn{4}{l}{\textbf{\textsc{Randomized Models}}}\\
\multicolumn{1}{c}{\textbf{\textsc{Model}}} & \multicolumn{1}{c}{\textbf{\textsc{Time}}} & \multicolumn{1}{c}{\textbf{\textsc{Energy}}} 			& \multicolumn{1}{c}{\textbf{\textsc{Notes}}}\\\cline{1-4}
			 & $O(n\log\Delta \log^2 n)$ 						& $O(\log\Delta \log^2 n)$ & $$ \istrut[2]{4}\\\cline{2-4}
$\nocd$ 		& $O(D^{1+\epsilon} \log^{O(1/\epsilon)} n)$ 	& $O(\log^{O(1/\epsilon)}n)$		 						&  $\epsilon>0$ \istrut[2]{4}\\\cline{2-4}
            & $O(n \log n)$ 										& $O(\log n)$ 								& $\Delta = O(1)$ \istrut[2]{4}\\\cline{2-4}
			& any 										& $\Omega(\log\Delta \log n)$ 								& \cite{Newport14}, bipartite graph $K_{2,k}$, $1\leq k\leq \Delta$ \istrut[2]{4}\\\hline
			& $O\left(\frac{n \log \Delta \log^{2+\epsilon} n}{\epsilon \log \log n}\right)$ & $O\left(\frac{\log^2 n}{\epsilon \log \log n}\right)$ 	&  $\epsilon \in (0,1)$ \istrut[2]{4}\\\cline{2-4}
$\cd$ 		& $O\left(\Delta n^{1+\xi}\right)$  					& $O\left(\frac{\log n (\log \log \Delta + \xi^{-1})}{\log \log \log \Delta}\right)$ &  $\xi = \omega(\log \log n / \log n)$ \istrut[2]{4}\\\cline{2-4}
	 		& any										& $\Omega(\log n)$ 										& \cite{Newport14}, bipartite graph $K_{2,k}$, $1\leq k\leq \Delta$\istrut[2]{4} \\\hline
		 	& $O(n\log n)$ 									& $O(\log n)$ 											& $$\istrut[2]{4} \\\cline{2-4}
\rb{2.6}{$\LOCAL$} & any										& $\Omega(\log n)$ 										& Path graph \istrut[2]{4}\\\hline
\multicolumn{4}{l}{}\\
\multicolumn{4}{l}{\textbf{\textsc{Deterministic Models}}}\\\cline{1-4}
$\nocd$ 		& any										& $\Omega(\Delta)$ 										& \cite{JurdzinskiKZ02}, bipartite graph $K_{2,k}$, $1\leq k\leq \Delta$ \istrut[2]{4}\\\hline
$\cd$ 		& $O(N^2 n\log n \log N)$ 								& $O(\log^3 N\log n)$ 									& $$\istrut[2]{4} \\\hline
$\LOCAL$ 	& $O(n\log n \log N)$								 & $O(\log n \log N)$ 		 								& $$ \istrut[2]{4}\\\hline
\end{tabular}
\caption{\label{table:broadcast}A summary of our results.  We are aware of no prior work on the energy complexity of $\broadcast$.
Parameters: $n$ is the number of vertices, $\Delta$ the maximum degree, $D$ the diameter, and $\{1,\ldots,N\}$ the ID space;
some algorithms are parameterized by constants $\epsilon$ and $\xi$.}
\end{table}

\subsection{Related Work}
Energy saving is a critical issue for sensor networks, and it has attracted a lot of attentions in networking and systems research.
Most commercial devices in a sensor network, such as MICAz and SunSPOT, can switch between {\em active} and {\em sleep} modes~\cite{Taniguchi12}; the energy consumption of a device in sleep mode is significantly smaller than in active mode. In~\cite[Section 9.1]{AKYILDIZ2007921},
{\em idle listening} (i.e.,  a device is active, but no message is received) and {\em packet collisions} are identified as major causes of
energy loss. An approach to this issue is to adaptively set the work/sleep cycle of the devices~\cite{vanDam2003,Ye2004,Wang2006}; based on this approach, practical energy-efficient algorithms for $\broadcast$ have been designed~\cite{Hong2009,Hong13}. 
Another route to reducing the energy cost is via {\em Time Division Multiple Access}
(TDMA) algorithms, which reduce collisions by properly assigning time slots to the devices~\cite{Herman2004,Kulkarni2004}.

\medskip

Despite its importance in practice, energy complexity has not received much study in theory research.
Most prior work that measured energy complexity/channel accesses considered only {\em single-hop} networks.
Nakano and Olariu~\cite{NakanoO00} showed that in $\nocd$, $n$ initially identical devices
can assign themselves distinct IDs in $\{1,\ldots,n\}$ with $O(\log\log n)$ energy per device.
Bender, Kopelowitz, Pettie, and Young~\cite{BenderKPY16} gave a randomized method for
circuit-simulation in the $\cd$ model, which led to algorithms for $\leader$ and $\approxcounting$
using $O(\log(\log^* n))$ energy and $n^{o(1)}$ time, w.h.p.
An earlier algorithm of Kardas et al.~\cite{kardas2013energy} solves the problem in
$O(\log^\epsilon n)$ time using $O(\log\log\log n)$ energy, but only in expectation.
Chang et al.~\cite{ChangKPWZ17} proved that for these problems,
$\Theta(\log(\log^* n))$ and $\Theta(\log^* n)$ energy are optimal in $\cd$ and $\nocd$, respectively,
for $\poly(n)$-time algorithms.
They also give tradeoffs between time and energy, e.g., in $\nocd$, with $O(\log^{2+\epsilon} n)$ time we can use just $O(\epsilon^{-1}\log\log\log n)$ energy.
For deterministic $\leader$ protocols, $\Theta(\log N)$ is optimal in $\cd$ and $\nocd$~\cite{ChangKPWZ17,JurdzinskiKZ02c},
but if senders can also detect collisions, the energy complexity drops to $\Theta(\log\log N)$~\cite{ChangKPWZ17}.
See also~\cite{JurdzinskiKZ02b,JurdziskiKZ03,JurdzinskiKZ02c}.

\medskip

$\broadcast$ is a well-studied problem in multi-hop networks, but nearly all prior research focused solely on \emph{time} complexity.
The seminal \emph{decay} algorithm of Bar-Yehuda et al.~\cite{bar1992time} solves $\broadcast$ in $\nocd$
in $O(D\log n + \log^2 n)$ time.  This bound was later improved to $O(D \log \frac{n}{D} + \log^2 n)$~\cite{CR06,KP05}.
The $\log^2 n$ term is known to be necessary~\cite{ABLP91}, and the $D\log\frac{n}{D}$ term is known to be optimal~\cite{KM98} for a restricted
class of algorithms that forbid ``spontaneous transmission.''\footnote{Vertices that have yet to learn the message are forbidden from transmitting.}

Haeupler and Wajc \cite{haeupler2016faster} recently gave an $O(D \frac{\log n \log \log n}{\log D} + \log^{O(1)}n)$-time
$\broadcast$ algorithm in $\nocd$, demonstrating that spontaneous transmissions \emph{are} useful.
Czumaj and Davies~\cite{CzumajD17} improved the bound to $O(D \frac{\log n}{\log D} + \log^{O(1)}n)$
and gave a $\leader$ algorithm of the same complexity, improving~\cite{GhaffariH13}.
See~\cite{GhaffariHK15} for an $O(D + \log^6 n)$-time $\broadcast$ algorithm in the $\cd$ model.

\subsection{Organization and Technical Overview}

In Section~\ref{sect:lowerbound} we show two simple lower bounds.
We prove that even for a simple network topology---a path---and the strongest model---randomized $\LOCAL$---the
$\broadcast$ problem still requires  $\Omega(\log n)$ energy. 
We then present a generic reduction showing that the \emph{energy} complexity of $\broadcast$ 
in a multi-hop network is at least the \emph{time} complexity of $\leader$ in a single-hop network,
with the other aspects of the model being the same ($\cd$ or $\nocd$, deterministic or randomized).  
The take-away message from these lower bounds is that
the cost of $\broadcast$ arises from two causes: (i) the cost of synchronization, for propagating
messages along long paths (when $D$ is large), and (ii) the cost of contention-resolution in a vertex's 2-hop neighborhood (when $\Delta$ is large).

In Section~\ref{sect:localsim} we prove a general simulation theorem showing that any algorithm for the 
$\LOCAL$ model may be simulated in the $\cd$ and $\nocd$ models, with some blow-up in time and energy costs.
In Section~\ref{sect:bb} we introduce the basic building block $\sr$ used by all our algorithms.
In Section~\ref{section:rand} we show a simple randomized algorithm in $\nocd$ based on iterative clustering.
Our algorithm can be viewed as a mutual speed-up procedure. On the one hand, maintaining a clustering help 
conserve energy for broadcast. 
On the other hand, we use broadcast to get a better clustering. 
For graphs of unbounded degree, the energy cost of our algorithm is $O(\log^3 n)$ in $\nocd$, which is actually the product
of our two lower bounds.  Its runtime is $O(n\log^3 n)$.

In Section~\ref{section:diameter}, we improve the runtime of our randomized algorithms to
$O(D^{1+\epsilon} \poly(\log n))$. Our algorithm offers a continuous tradeoff between time and energy.
For any $\epsilon>0$, $\broadcast$ is solved using $O(\log^{O(\frac{1}{\epsilon})}n)$ energy
in $O(D^{1+\epsilon} \log^{O(\frac{1}{\epsilon})} n)$ time.

In the randomized $\cd$ model, we can almost achieve our energy lower bound for $\broadcast$
by spending more time.  In Section~\ref{sect:improve-cd} we show that for graphs of unbounded degree,
our algorithm uses $O(\frac{\log n\log\log \Delta}{\log\log\log \Delta})$ energy and
takes $O(\Delta n^{1+\xi})$ time, for any constant $\xi>0$.

In Section~\ref{section:path}, we present an algorithm for $\broadcast$ on a path,
and prove it has nearly optimal performance in terms of time and energy use.

In Appendix~\ref{section:det} we present a deterministic algorithm for the $\cd$ model,
which is also based on the idea of iterative clustering. We use \emph{ruling sets} to build each clustering. 
The energy complexity is $O(\log^3 N \log n)$ but the runtime is $O(nN^2 \cdot \log n \log N)$.

\section{$\broadcast$ Lower Bounds}\label{sect:lowerbound}

In this section we prove two lower bounds on the energy-complexity of $\broadcast$.


\begin{theorem}\label{thm:lb-path}
Consider an $n$-vertex path graph $P=(v_1, \ldots, v_n)$, where each vertex $v_i$ does not know its position $i$.
Suppose that $v_1$ attempts to broadcast a message $m$.
For any randomized $\LOCAL$ $\broadcast$ algorithm $\mathcal{A}$, with probability $1/2$,
at least one vertex spends $\frac15\log n$ energy before receiving the message $m$.
\end{theorem}

\begin{proof}
Without loss of generality, we assume the algorithm $\mathcal{A}$ works as follows. Initially each vertex $v$ generates a random string $r_v$ from the same distribution. During the execution of $\mathcal{A}$, each vertex $v_i$ maintains an interval $[\alpha, \beta]$ such that $v_i$ knows all random strings $r_{\alpha}, \ldots, r_{\beta}$ (initially $\alpha=\beta=i$).  Each time a vertex $v_i$ wakes up it either (i) sends all random strings it has to its two neighbors, or (ii) listens to the channel to receive random strings from its two neighbors. After each wakeup, $v_i$ decides, based on all information it has, when to wake up next and whether to send or listen.

Let $I$ be an interval of $(v_1,\ldots,v_n)$.
Define the event $E_i[I]$ as follows.  Suppose (contrary to reality) that the vertices outside of $I$
take the most advantageous actions to maximize the probability that all $I$-vertices learn the random string of
some vertex outside $I$.\footnote{For example, they may transmit in \emph{every} round.  We place no energy constraint on their behavior.}
Even given this help, there exists some
vertex in $I$ that, after its $i$th wake-up, only knows random strings of vertices in the interval $I$.
If $E_i[I]$ occurs, we write $v^\star[I]$ to denote the rightmost vertex in $I$ satisfying the statement above.
Notice that $E_i[I]$ depends solely on random strings of vertices in $I$, and thus $E_i[I_1]$ and $E_i[I_2]$ are independent for any two disjoint intervals $I_1$ and $I_2$.
We prove by induction that for any interval $I$ with $L_i \bydef (32)^i$ vertices,
$E_i[I]$ happens with probability at least $1/2$,
and this immediately implies the desired $\log_{32} n=\frac15\log n$ energy lower bound.

The base case of $i=0$ is trivial. Suppose the claim holds for intervals of length $L_i$.
Let $I$ be an interval of length $L_{i+1}$, partitioned into 32 intervals
$(I_1$, $\ldots$, $I_{32})$ with length $L_{i}$.
By the induction hypothesis, $E_{i}[I_j]$ happens with probability at least $1/2$.
By the independence of the intervals,
the probability that $E_i$ happens on at least $11$ of them is at least
$\Prob[\text{Binomial}(32,1/2) \geq 11] > 0.97 > \frac34$.
Conditioning on this event happening, we denote those 11 intervals as $J_1$, $\ldots$, $J_{11}$, from left to right.
Let $t_s$ be the time of the $(i+1)$th wake-up of $v^\star[J_s]$.
In order for $v^\star[J_6]$ to receive some random string of a vertex outside $I$ during its
$(i+1)$th wake-up, we need $t_6$ to be the largest number among either $\{t_1,\ldots.t_6\}$ or $\{t_6,\ldots,t_{11}\}$.
By symmetry and independence, $\{t_s\}_{1 \leq s \leq 11}$ are i.i.d.~random variables.
Thus, the probability that $v^\star[J_6]$ receives information from outside $I$ during its $(i+1)$th
wake-up is at most $\frac16+\frac16=\frac13$,
and so the probability of $E_{i+1}[I]$ happening is at least $\frac34\times\frac23=\frac12$.
This confirms the inductive hypothesis at $i+1$.
\end{proof}

Next, we prove $\broadcast$ lower bounds for the $\nocd$ and $\cd$ models,
which hold even in constant diameter graphs.

\begin{theorem}\label{thm:lb-via-leaderelection}
We have the following energy lower bounds for $\broadcast$:
(i) deterministic $\nocd$: $\Omega(\Delta)$,
(ii) randomized $\nocd$: $\Omega(\log\Delta \log n)$,
(iii) randomized $\cd$: $\Omega(\log n)$.
\end{theorem}
\begin{proof}
Consider the $\leader$ problem in a single-hop network, where the number of vertices is unknown, but is guaranteed to be at most $n'$.
Suppose that solving this problem with probability at least $1-f$ requires $T(n',f)$ time, even if the vertices have {\em shared randomness}, and they are allowed to send and listen simultaneously (the \emph{full duplex} model.)

Let $G_k=(\{s, v_1, \ldots, v_k, t\}, \{(s,v_1),\ldots (s,v_k), (t,v_1),\ldots (t,v_k)\})$ be isomorphic to $K_{2,k}$,
where the vertex $s$ attempts to broadcast a message.
We claim that  any $\broadcast$ algorithm that applies to the graphs $G_k$, for all $1\leq k\leq \Delta$,
with failure probability $f$, has an energy lower bound of $T(\Delta,f) / 2$.

The claim is proved by the following generic reduction.
Let $\mathcal{A}$ be any such $\broadcast$ algorithm that takes $E$ energy.
We transform it to a $\leader$ algorithm $\mathcal{A}'$ in a single-hop network that takes $2E$ time.
The idea is to treat $\{v_1, \ldots, v_k\}$ as the vertices of a single-hop network,
and treat $\{s,t\}$ as the communication channel.
Notice that the two vertices $s$ and $t$ do not know anything beyond (i) the local random bits in $s$ and $t$, and (ii) the feedback from the communication channel so far.
Therefore, each vertex in $\{v_1, \ldots, v_k\}$ can perfectly predict the future actions of the two vertices $s$ and $t$, if (i) all vertices in $\{v_1, \ldots, v_k\}$ know the local random bits generated by $s$ and $t$, and (ii) they always listen to the channel.
Hence we may simulate the algorithm $\mathcal{A}$ on a single-hop network by using the shared randomness to simulate the local random bits generated by $s$ and $t$.
Any time slot where both $s$ and $t$ are not listening is meaningless, so we can skip it.
Therefore, the simulation takes at most $2E$ time.
Notice that in order to have $t$ receive the message from $s$, there must be one time slot where
exactly one vertex in $\{v_1, \ldots, v_k\}$ transmits.  This is precisely the termination condition of
a $\leader$ algorithm in the full duplex model.  Hence, we obtain the the desired $\leader$
algorithm $\mathcal{A}'$.

Notice that the above reduction works in both
the $\cd$ and $\nocd$ models.
It has been shown in~\cite{Newport14} that
$T(n',f) = \Omega(\log\log n'+\log \frac1f)$
in randomized $\cd$, and
$T(n',f) =\Omega(\log n' \log \frac1f)$ in randomized
$\nocd$, with or without full duplex.
Thus, we obtain the two desired
energy lower bounds for $\broadcast$ in the randomized model.

For the deterministic model,
an $\Omega(N)$ time lower bound has been shown in~\cite{JurdzinskiKZ02}.
In their setting, the size of the single-hop network is unknown but is at most $N$, and each vertex has a distinct ID in $\{1, \ldots, N\}$.
This lower bound, together with the above generic reduction, implies the $\Omega(\Delta)$ energy lower bound for $\broadcast$ in the deterministic model. Some minor modifications are needed. We let the IDs of the vertices in $\{v_1, \ldots, v_k\}$ be chosen from the range $\{1, \ldots, \Delta\}$, and let the IDs of $s$ and $t$ be $\Delta+1$ and $\Delta+2$. Since the parameter $\Delta$ is common knowledge, the IDs of $s$ and $t$ are known to all vertices initially. Hence each vertex in $\{v_1, \ldots, v_k\}$ can perfectly predict the future actions of the two vertices $s$ and $t$ solely according to the channel feedback.
\end{proof}

Theorem~\ref{thm:lb-via-leaderelection} complements
Theorem~\ref{thm:lb-path} by showing another
$\Omega(\log n)$ energy lower bound (by setting $f=1/\poly(n)$) in $\cd$, even when $D=O(1)$.
On graphs with unbounded degree,
Theorem~\ref{thm:lb-via-leaderelection} implies
$\Omega(\log^2 n)$ energy lower bounds in $\nocd$,
and $\Omega(n)$ lower bounds in \emph{deterministic} $\nocd$.



\section{Simulation of $\LOCAL$ Algorithms}\label{sect:localsim}

In this section, we show that with a preprocessing step, it is possible
to simulate any $\LOCAL$ algorithm in $\nocd$ (and therefore also in $\cd$) by scheduling all transmissions to avoid collisions.
There is a cost in both time and energy to run the simulation, which makes it most efficient when $\Delta$ is constant.

\begin{theorem}\label{thm:simlocal}
Any algorithm $\mathcal{A}$ for the $\LOCAL$ model taking time $T$ and energy $E$ 
can be simulated in randomized $\nocd$ using $O(\Delta^2 T + \Delta \log \Delta \log n)$ time and $O(\Delta (E + \log \Delta \log n))$ energy.
This result holds even when $\mathcal{A}$ is full-duplex, but simultaneous transmission and reception are not allowed in the simulation.
\end{theorem}

The main idea of the simulation is that if we were given a $k$-coloring of $G$ with the property that for every vertex $v$, the vertices in $N^{+}(v) = N(v) \cup \{v\}$ are all distinct colors (i.e., it is a $k$-coloring of $G+G^2$), then we could divide time up into blocks of length $k$, each block representing a single time step of the $\LOCAL$ model.
A vertex with color $j$ would transmit only in the $j$th time step of any block, and listen in only the slots corresponding to its neighbors' colors. Then the property of the coloring ensures that no two vertices that are within distance $2$ will ever transmit in the same time step, thus eliminating collisions altogether. This increases the complexity of the simulated algorithm by a factor $k$ for time and a factor $\Delta$ for energy.

In what follows, we show how to generate a coloring with $k=2\Delta^2$ in a distributed manner, using $O(\Delta \log \Delta \log n)$ time and energy in the $\nocd$ model.

\subsection{Distributed Coloring of $G+G^2$}

We assume that at the beginning, each vertex has an $O(\log n)$-bit distinct ID.
For the purpose of generating the coloring, the vertices will each need to know their own degree.

\paragraph{Algorithm: {\sf Learn-degree.}} For $C\Delta \log n $ time steps, independently in each time step, each vertex $v$ sends 
$\ID(v)$ with probability $1/\Delta$; otherwise it listens.

\begin{lemma}
By the end of {\sf Learn-degree}, with high probability all the vertices learn their degree and the IDs of all their neighbors.
\end{lemma}
\begin{proof}
The proof is via a coupon collector tail bound.
Consider a vertex $v$.  Let $w$ be a neighbor of $v$. At any time step, since $v$ listens with probability
$1-1/\Delta$, each neighbor speaks independently with probability $1/\Delta$ and and exactly one neighbor must speak in order for $v$ to hear, the probability that $v$ learns the ID of $w$ is  $(1-1/\Delta)^{\deg(v)}/\Delta   \ge 1/4\Delta$.
Thus the probability that $v$ has not heard from $w$ after $T$ time steps is at most  $(1-1/4\Delta)^T \le e^{-\frac{T}{4\Delta} }$. Taking  union bounds over all neighbors of all vertices, the probability that there is a pair of neighbors $v$ and $w$ such that  $v$ has not heard from $w$ after $T$ steps is at most  $ n\Delta e^{-\frac{T}{4\Delta} }$. Since the algorithm runs for $C\Delta \log n$ time steps, as long as $C >4$, it succeeds with  high probability.
\end{proof}

Now we return to the problem of generating the desired coloring. We may assume that each vertex already knows its own degree, and that of each of its immediate neighbors. The {\sf Two-Hop-Coloring} algorithm runs for $C \log n$ iterations, where each iteration is described below.

\paragraph{Algorithm: {\sf Two-Hop-Coloring}, Single Iteration.} Each vertex $v$ does the following steps in parallel.
\begin{enumerate}
\item If the color of $v$, $c(v)$, was fixed in a previous iteration, it remains unchanged.
Otherwise, randomly sample a new \emph{proposed color} $c(v) \in [2\Delta^2]$.

\item Vertex $v$ maintains a vector, $L(v)$, of the most recently announced color for each of his neighbors in $G$.  Initially,
the entries in this list are all ``undefined.''  During the protocol, $v$ will announce $L(v)$, together with the label $\ID(v)$, 
to her neighbors, at random times chosen at rate $1/\Delta$.  Vertex $v$ will also maintain for her own records, 
a copy, for each neighbor $w$, of the most recently heard version of $L(w)$.

\item \label{step:xxxx} For $C\Delta \log \Delta$ time steps, independently in each time step, 
\begin{itemize}
\item[$\bullet$] with probability $1/\Delta$, $v$ sends $(\ID(v), c(v), L(v))$;
\item[$\bullet$] otherwise, $v$ listens. If $v$ hears $(\ID(w), c(w), L(w))$ from a neighbor $w$, she uses $c(w)$ to update $L(v)$ and updates her local record of $L(w)$ to match the message. 
\end{itemize}
\item \label{step:xxxxx} Suppose $v$ has yet to permanently fix $c(v)$. 
The current candidate will be rejected if either of the following conditions hold: 
\begin{enumerate}
\item[(i)] some entry of $L(v)$ equals $c(v)$ or is undefined, or
\item[(ii)] for some neighbor, $w$, of $v$, some entry of $L(w)$ is undefined, or at least two entries of $L(w)$ equal $c(v)$.
\end{enumerate}
Otherwise, $v$ permanently colors itself $c(v)$, confident that no other vertex within distance two of $v$ in $G$ has chosen the same color.
\end{enumerate}

\begin{lemma}\label{lem:two-hop-sr}
A single iteration of {\sf Two-Hop-Coloring} results in vertex $v$ having fixed its color with constant probability.
\end{lemma}
\begin{proof}
Once more we prove this via a coupon collector tail bound. Fix vertex $v$ and its neighbor $w \in N(v)$. As before, the probability that $v$ has not heard from $w$ in $T$ time steps (of Step~\ref{step:xxxx}) is less than $e^{-\frac{T}{4\Delta} }$.  Upon hearing from a neighbor, $v$ learns that neighbor's proposed color. Taking a union bound over the neighbors of $v$, the probability that in $T$ time steps $v$ has not learned the colors of all of its neighbors is at most
$\Delta e^{-\frac{T}{4\Delta} }$. Now, taking another union bound over $N^+(v)$, the probability that after $T$ steps, there is some vertex in $N^+(v)$ who has not learned the colors of all their neighbors is at most $\Delta (\Delta+1) e^{-\frac{T}{4\Delta}}$. Thus, within  $O(\Delta \log\Delta)$ steps, with probability at least $1/2$, \emph{everyone in} $N^+(v)$ has learned all the colors in their neighborhood. Another $O(\Delta \log\Delta)$ steps allows $v$ to hear from each of its neighbors one more time 
(with constant probability), thereby learning all the colors in its distance-2 neighborhood.

Assume that $v$ has learned all of $\{c(w) \;|\; w \in N^2(v)\}$. Observe that $|N^2(v)| < \Delta^2$, and  $c(v)$ was sampled uniformly at random from $2\Delta^2$ colors. Thus,  $c(v)$ is different from all of  $\{c(w) \;|\; w \in N^2(v)\}$ with probability at least $1/2$. In this event, $v$ permanently fixes its color $c(v)$ in Step~\ref{step:xxxxx}.
\end{proof}

\begin{lemma}
Algorithm {\sf Two-Hop-Coloring} produces a proper coloring of $G+G^2$ with high probability.
\end{lemma}

\begin{proof}
Lemma~\ref{lem:two-hop-sr} indicates that the number of iterations until $v$ fixes its color is a geometric random variable with constant mean.  Since we run for $\Omega(\log n)$ iterations, it follows that $v$ has fixed its color with
high probability. Moreover, a union bound shows that in $\Omega(\log n)$ iterations (with a bigger constant) \emph{all} vertices in $G$ have fixed their color with high probability. Thus, with high probability, a coloring has been generated. It is straightforward to see that the resulting coloring is a proper coloring of $G+G^2$ (in view of Step~\ref{step:xxxxx}).
\end{proof}

\section{Basic Building Blocks}\label{sect:bb}

Let $S$ and $R$ be two disjoint vertex sets.
In the $\sr$ task,
each vertex $u \in S$ attempts to transmit a message $m_u$,
and each vertex in $R$ attempts to receive one message.
An $\sr$ algorithm guarantees that for
every $v \in R$ with $N(v) \cap S \neq \emptyset$,
with probability $1-f$, $v$ receives a message
$m_u$ from \emph{some} vertex $u \in N(v) \cap S$.

\begin{lemma}\label{lemma:1}
In the randomized $\nocd$ model, $\sr$ can be solved with high probability, i.e., $f = 1/\poly(n)$, in time $O(\log \Delta \log n)$ and energy $O(\log \Delta \log n)$.
\end{lemma}
\begin{proof}
Use the $O(\log \Delta \log 1/f)$-time algorithm of~\cite{bar1992time}, which is also known as {\em decay}.
\end{proof}

\paragraph{A Generic Transformation.}
Suppose that there is an algorithm $\mathcal{A}$ which elects a leader in time $T(n',f)$ with probability $1-f$ in a single-hop network, where the number of vertices is unknown, but is guaranteed to be at most $n'$. We assume that the algorithm $\mathcal{A}$ and the underlying single-hop network have the following properties.
\begin{itemize}
\item We allow the vertices in the single-hop network to simultaneously send and listen. Since we do not measure the energy of $\mathcal{A}$, we can assume that all vertices (including the ones that are transmitting) are always listening to the channel. Thus, a leader is elected once a message is successfully sent.
\item We assume that algorithm $\mathcal{A}$ is {\em uniform} in the following sense. For each time slot $t$, there is a number $k_t \in \{0, 1, \ldots, \lceil \log n \rceil\}$ such that each vertex transmits with the same probability $p=2^{-k_t}$ independently at the time slot $t$. The number $k_t$ depends only on the history of the algorithm execution
before time $t$. Since all vertices are always listening to the channel, they have the same information.
\end{itemize}

We claim that there is a randomized algorithm $\mathcal{A'}$ that solves $\sr$ in time $T(\Delta,f) \cdot \lceil \log \Delta \rceil$ with energy cost $2 \cdot T(\Delta,f)$, but in a multi-hop network, where vertices cannot simultaneously
send and listen.
The algorithm $\mathcal{A'}$ consists of $T(\Delta,f)$ epochs, each of which consists of $\lceil \log \Delta \rceil$ time slots. Each vertex $v \in S$ transmits at the $i$th time slot with probability $2^{-i}$ in such a way that the total number of transmissions of $v$ during an epoch is at most $2$ (since $1 + \frac12 + \frac14 + \frac18 + \ldots = 2$). Each vertex $u \in R$ simulates the algorithm $\mathcal{A}$ as follows. During the $i$th epoch, $u$ only listens at the $k_i$th slot; and $u$ calculates $k_{i+1}$ based on the information it receives from the channel so far. By the last epoch, each vertex $u \in R$ receives a message with probability $1-f$ (since $|N(u) \cap S|\leq \Delta$). Based on this generic transformation, we obtain Lemma~\ref{lemma:1-mod}.

\begin{lemma}\label{lemma:1-mod}
In the randomized $\cd$ model, $\sr$ can be solved with energy $O(\log \log \Delta + \log 1/f)$ and runtime $O(\log \Delta (\log \log \Delta + \log 1/f))$. For the special case where each $v \in S$ is adjacent to at most one vertex in $R$, the energy cost is $O(\log\log \Delta) + X$, where $X$ is a random variable drawn from an exponential distribution $\operatorname{Exponential}(\lambda)$, for some $\lambda = O(1)$.
\end{lemma}
\begin{proof}
Apply the above generic transformation to the $O(\log\log n' + \log 1/f)$-time uniform leader election algorithm of~\cite{nakano2002}.
The runtime of this algorithm is $O(\log\log n') + X$, where $X$ is a random variable drawn from an exponential distribution $\operatorname{Exponential}(\lambda)$, for some $\lambda = O(1)$. Thus, to have failure probability at most $f$, it needs $O(\log\log n' + \log 1/f)$ time.
For the special case where each $v \in S$ is adjacent to at most one vertex in $R$, consider the following modifications.
\begin{itemize}
 \item In the first round, all vertices in $R$ speak, and all vertices in $S$ listen. This allows each vertex in $S$ to check whether it is adjacent to a vertex in $R$. Those vertices in $S$ that are not adjacent to a vertex in $R$ terminates after the first round.
 \item We allocate an additional time slot at the end of each epoch (in the generic transformation) to let each vertex $v$ in $R$ inform all its neighbors in $S$ whether $v$ has received a message. If $v$ has received a message, then all vertices in $S \cap N(v)$ can terminate.\qedhere
\end{itemize}
\end{proof}

\begin{remark}\label{rem:cd}
In  Lemma~\ref{lemma:1-mod}, if a vertex $v$ satisfies either (i) $v \in S$ and $N(v) \cap R = \emptyset$, or (ii) $v \in R$ and $N(v) \cap S = \emptyset$, then the energy cost of $v$ can be lowered to $O(1)$ in the $\cd$ model. Due to the ability of a vertex to distinguish between noise and silence, in $O(1)$ time, each $v \in S$ can check whether $N(v) \cap R = \emptyset$, and similarly each $v \in R$ can check whether $N(v) \cap S = \emptyset$ in $O(1)$ time. We will make use of this observation to reduce the energy cost of algorithms in the $\cd$ model.
\end{remark}

\section{Basic Energy-Efficient Randomized Algorithms}\label{section:rand}

In this section we design energy-efficient algorithms for the $\broadcast$ problem in multi-hop networks.
In particular, we show that $\broadcast$ can be solved using $O(\log^3 n)$ energy in randomized $\nocd$.

\paragraph{Layers of Vertices.} A labeling $\mathcal{L}: V(G) \mapsto \{0, \ldots, n-1\}$ is said to be {\em good} if it has the following property. Each vertex $v$ with $\mathcal{L}(v) > 0$ must have a neighbor $u$ such that $\mathcal{L}(u)=\mathcal{L}(v)-1$. With respect to a good labeling $\mathcal{L}$, a vertex $v$ is called a {\em layer-$i$} vertex if $\mathcal{L}(v) = i$. The intuition underlying the definition of a good labeling is that it represents a clustering of vertices. If we let each layer-$i$ vertex select a layer-$(i-1)$ neighbor as its parent, then we obtain a partition of $V(G)$ into
$|\mathcal{L}^{-1}(0)|$ clusters.
Each cluster $C$ is a rooted tree $T$, where the root $r$ is the unique layer-0 vertex in the cluster $C$. However, it is possible that a vertex has multiple choices of its parent, so the clustering resulting from a good labeling is, in general, not unique.

We say that two layer-0 vertices $u$ and $v$ are {\em $\mathcal{L}$-adjacent} if there exists a path $P = (u, u_1, \ldots, u_a$, $v_b, \ldots, v_1, v)$ such that $\mathcal{L}(u_i)=i$ for all $i \in [a]$ and $\mathcal{L}(v_j)=j$ for all $j \in [b]$. The graph $G_{\mathcal{L}}$ is on vertex set $\mathcal{L}^{-1}(0)$ and
edge set $\{\{u,v\} \;|\; \text{$u$ and $v$ are $\mathcal{L}$-adjacent}\}$.

In the following lemma we show that  $\broadcast$ can be solved
energy-efficiently if we already have a good labeling $\mathcal{L}^\star$
with small number of layer-0 vertices.

\begin{lemma}\label{lem:reduce}
Let $\mathcal{L}^\star$ be a good labeling of $G$.  Each vertex knows its $\mathcal{L}^\star$-label
and two integers $d,L \geq 1$ such that
(i) $d$ is an upper bound on the diameter of $G_{\mathcal{L}^\star}$,
and
(ii) $L$ is an upper bound of the number of layers.
Then $\broadcast$ can be solved by a randomized algorithm
with high probability in time $T(n,d,L)$ using energy $E(n,d,L)$.
\begin{align*}
\LOCAL:& & & T(n,d,L) = O(L d) & &E(n,d,L) = O(d)\\
\cd:& & & T(n,d,L) = O(L d \log n \log \Delta) & &E(n,d,L) = O(d + \log n)\\
\nocd:& & & T(n,d,L) = O(L d  \log n \log \Delta) & &E(n,d,L) = O(d \log n \log \Delta)
\end{align*}
\end{lemma}
\begin{proof}
Let $v$ be the vertex that attempts to broadcast some message $m$.
The goal of the $\broadcast$ problem is to relay the message $m$ to all vertices in the graph.
This can be solved by first (1) do {\sf Up-cast} to relay the message from $v$ to some layer-0 vertex; (2) repeat ({\sf Down-cast}, {\sf All-cast}, {\sf Up-cast}) for $d$ times to let all layer-0 vertices receive the message $m$; and then (3) do {\sf Down-cast} to ensure that all vertices in the graph have the message $m$.
\begin{itemize}
\item {\sf Down-cast.} For $i = 0, \ldots, L-2$, do $\sr$ with $S$ being the set of layer-$i$ vertices that hold the message $m$, and $R$ being the set of all layer-$(i+1)$ vertices that have not received the message $m$. Each vertex in $S$ attempts to broadcast the message $m$.
\item {\sf All-cast.} Do $\sr$ with $S$ being the set of all vertices that hold the message $m$, and $R$ being the set of all vertices that have not received the message $m$. Each vertex in $S$ attempts to broadcast the message $m$.
\item {\sf Up-cast.} For $i = L-1, \ldots, 1$, do $\sr$ with $S$ being the set of layer-$i$ vertices that hold the message $m$, and $R$ being the set of all layer-$(i-1)$ vertices that have not received the message $m$. Each vertex in $S$ attempts to broadcast the message $m$.
\end{itemize}

We use $\sr$ with $f = 1/\poly(n)$.
Thus, the $\broadcast$ problem can be solved in $O(L d) \cdot T'(n,\Delta)$ time and $O(d) \cdot E'(n,\Delta)$ energy, where $T'(n,\Delta)$ and $E'(n,\Delta)$ are the runtime and the energy cost of $\sr$; see Lemmas~\ref{lemma:1} and~\ref{lemma:1-mod}.
By the observations made in Remark~\ref{rem:cd},
the energy cost can be further reduced to $O(d + E'(n,\Delta)) = O(d + \log n)$.
In the above algorithm, each vertex $v$ is involved in $O(d)$ invocations
of $\sr$, and all but $O(1)$ of them satisfy either (i) $v \in S$ and
$N(v) \cap R = \emptyset$, or (ii) $v \in R$ and $N(v) \cap S = \emptyset$.
\end{proof}

In what follows, we show that a good labeling $\mathcal{L}^\star$ with small number of layer-0 vertices can be computed efficiently.
Our strategy is to begin with the trivial all-0 good labeling, and then in each iteration use the current good labeling 
$\mathcal{L}$ to obtain a new good labeling $\mathcal{L}'$ such that 
(i) each layer-0 vertex remains layer-0 with some probability (to be determined), 
and 
(ii) no new layer-0 vertex is created.

\paragraph{Computing a New Labeling $\mathcal{L}'$ from $\mathcal{L}$.} Let $p\in (0,1)$ and $s\geq1$ be two parameters to be chosen later.
The algorithm for computing  $\mathcal{L}'$  is as follows: 
(1) initially, set $\mathcal{L}'(v)=\bot$ for all vertices, but each layer-0 vertex $v$ sets $\mathcal{L}'(v)=0$ independently with probability ${p}$; 
(2) repeat ({\sf Down-cast}, {\sf All-cast}, {\sf Up-cast}) $s$ times, and then do {\sf Down-cast}; 
(3) any vertex $v$ that has yet to obtain a new $\mathcal{L}'$ label (i.e., $\mathcal{L}'(v)=\bot$) retains its old label: $\mathcal{L}'(v)=\mathcal{L}(v)$.

\begin{itemize}
\item {\sf Down-cast.} For $i = 0, \ldots, n-2$, do $\sr$ with $S$ being the set of layer-$i$ vertices of $\mathcal{L}$ such that its $\mathcal{L}'$ label is not $\bot$, and $R$ being the set of all layer-$(i+1)$ vertices of $\mathcal{L}$ such that its $\mathcal{L}'$ label is $\bot$. Each vertex in $S$ attempts to broadcast its $\mathcal{L}'$ label. Each vertex in $R$ that receives the message $m$ sets its $\mathcal{L}'$ label to be $m+1$.

\item {\sf All-cast.} Do $\sr$ with $S$ being the set of all vertices such that its $\mathcal{L}'$ label is not $\bot$, and $R$ being the set of all vertices such that its $\mathcal{L}'$ label is $\bot$. Each vertex in $S$ attempts to broadcast its $\mathcal{L}'$ label. Each vertex in $R$ that receives the message $m$ sets its $\mathcal{L}'$ label to be $m+1$.

\item {\sf Up-cast.} For $i = n-1, \ldots, 1$, do $\sr$ with $S$ being the set of layer-$i$ vertices of $\mathcal{L}$ such that its $\mathcal{L}'$ label is not $\bot$, and $R$ being the set of all layer-$(i-1)$ vertices of $\mathcal{L}$ such that its $\mathcal{L}'$ label is $\bot$. Each vertex in $S$ attempts to broadcast its $\mathcal{L}'$ label. Each vertex in $R$ that receives the message $m$ sets its $\mathcal{L}'$ label to be $m+1$.
\end{itemize}

We use $\sr$ with $f = 1/\poly(n)$.
It is straightforward to verify that the algorithm indeed computes a good labeling $\mathcal{L}'$, w.h.p.
The algorithm takes $O(ns) \cdot T'(n,\Delta)$ time and $O(s) \cdot E'(n,\Delta)$ energy, 
where $T'(n,\Delta)$ and $E'(n,\Delta)$ are the runtime and the energy cost of $\sr$; see Lemmas~\ref{lemma:1} and~\ref{lemma:1-mod}.
In the $\cd$ model, the energy cost is $O(s + \log n)$; see Remark~\ref{rem:cd}.

We show that each layer-0 vertex in $\mathcal{L}$ remains layer-0 in $\mathcal{L}'$ with probability at most $p+(1-p)^{\min\{s+1,w\}} + 1/\poly(n)$, where $w=|\mathcal{L}^{-1}(0)|$.
Assuming all invocations of $\sr$ succeed, which happens with probability $1-1/\poly(n)$,
there are two ways for a layer-0 vertex $v$ in $\mathcal{L}$ to remain layer-0 in $\mathcal{L}'$.
\begin{itemize}
\item The vertex $v$ sets $\mathcal{L}'(v)=0$ at Step~(1), and this occurs with probability $p$.
\item All vertices $u$ within distance $s$ to $v$ (in $G_{\mathcal{L}}$) have $\mathcal{L}'(u) = \bot$ at Step~(1),
and this occurs with probability at most $(1-p)^{\min\{s+1,w\}}$.
\end{itemize}
We are in a position to prove the main theorems of this section.

\begin{theorem}\label{theorem:3}
The $\broadcast$ problem can be solved by a randomized algorithm with high probability in the following runtime $T(n,\Delta)$ and energy cost $E(n,\Delta)$.
\begin{align*}
\LOCAL:& & & T(n,\Delta) = O(n \log n) & &E(n,\Delta) = O(\log n)\\
\cd:& & & T(n,\Delta) = O(n \log \Delta \log^2 n) & &E(n,\Delta) = O(\log^2 n)\\
\nocd:& & & T(n,\Delta) = O(n \log \Delta \log^2 n) & &E(n,\Delta) = O(\log \Delta \log^2 n)
\end{align*}
\end{theorem}

\begin{proof}
Set $p=1/2$ and $s=1$. As long as the number of layer-0 vertices in $\mathcal{L}$ is greater than 1,
each layer-0 vertex in $\mathcal{L}$ remains layer-0 in $\mathcal{L}'$ with probability at most
$p+(1-p)^{\min\{s+1,w\}} + 1/\poly(n) \leq 1/2 + 1/4 + 1/\poly(n) = 3/4 + 1/\poly(n)$.
Thus, after $O(\log n)$ iterations of computing a new labeling from an old labeling,
we obtain a good labeling $\mathcal{L}^\star$ such that the number of layer-0 vertices is exactly 1,
with high probability.
Applying Lemma~\ref{lem:reduce} (with $L=n$ and $d=0$) gives the theorem.
\end{proof}

Recall that the energy cost for computing $\mathcal{L}'$ from $\mathcal{L}$ is
$O(s + \log n)$ (instead of $O(s \log n)$) in the $\cd$ model.
Using this fact, the energy cost can be improved in the $\cd$ model without affecting the time too much.

\begin{theorem}\label{theorem:improved-1}
In the $\cd$ model, $\broadcast$ can be solved by a randomized algorithm with high probability in 
$O\left(\frac{n \log \Delta \log^{2+\epsilon} n}{\epsilon \log \log n}\right)$ time with 
energy cost $O\left(\frac{\log^2 n}{\epsilon \log \log n}\right)$, for any $\epsilon\in(0,1)$.
\end{theorem}

\begin{proof}
Set $p=\log^{-\epsilon/2} n$ and $s=\log^{\epsilon} n$. As long as the number of layer-0 vertices in $\mathcal{L}$ is greater than $\log^{\epsilon} n$, each layer-0 vertex in $\mathcal{L}$ remains layer-0 in $\mathcal{L}'$ with probability at most $p+(1-p)^{\min\{s+1,w\}} = O(\log^{-\epsilon/2}n)$. 
Thus, after $O\left(\frac{\log n}{\epsilon \log \log n}\right)$ iterations of computing new labeling from old labeling, we obtain a good labeling $\mathcal{L}^\star$ such that the number of layer-0 vertices is at most $\log^{\epsilon} n$, with high probability.
Notice that the energy cost of each iteration is $O(s + \log n) = O(\log n)$.
Applying Lemma~\ref{lem:reduce} (with $L=n$ and $d=\log^{\epsilon} n$) gives the theorem.
\end{proof}

By Theorem~\ref{thm:simlocal}, we can simulate the $\LOCAL$ algorithm of Theorem~\ref{theorem:3} in $\nocd$ with $\poly(\Delta)$ 
overhead in time and energy, and thereby provide an overall improvement (i.e., in both time and energy) 
for graphs with $\Delta = o(\sqrt{\log n \log \log n})$,
and an improvement in energy at the expense of time all the way up to $\Delta = o(\log n)$.
In particular, we have the following corollary, which shows that the $\LOCAL$ lower bound on path graphs (Theorem~\ref{thm:lb-path}) is matched by a $\nocd$ algorithm on bounded-degree graphs.

\begin{corollary}
In the $\nocd$ model, for bounded degree graphs, the $\broadcast$ problem can be solved by a randomized algorithm with high probability in $O(n \log n)$ time with energy cost $O(\log n)$.
\end{corollary}

\section{An $\tilde{O}(D^{1+\epsilon})$-Time $\broadcast$ Algorithm}\label{section:diameter}

In this section, we show that it is possible to achieve near diameter time $O(D^{1+\epsilon} \poly(\log n))$
while keeping relatively low energy complexity $O(\poly (\log n))$.
Throughout this section we are working in the $\nocd$ model for simplicity.  A couple log factors
can be saved by adapting our algorithm to the $\cd$ model.

Our algorithm is based on the following subroutine {\sf Partition($\beta$)},
described by Miller, Peng, and Xu~\cite{miller2013parallel} and further analyzed by
Haeupler and Wajc~\cite{haeupler2016faster}.
The goal of {\sf Partition($\beta$)}
is to produce the following random clustering.  Each vertex $v$ picks
$\delta_v\sim \text{Exponential}(\beta)$, $\beta \in (0,1)$,
and assigns $v$ to the cluster of $u$ that minimizes $\dist(u,v) - \delta_u$.
This algorithm can be implemented in $\nocd$ as follows~\cite{haeupler2016faster}.

\begin{description}
\item{\sf Partition($\beta$)}
Every vertex $v$ picks a value $\delta_v\sim \text{Exponential}(\beta)$.
Let $v$'s start time be $\text{start}_v\gets \frac{2\log n}{\beta}-\lceil \delta_v\rceil$.
There are $\frac{2\log n}{\beta}$ epochs numbered 1 through $\frac{2\log n}{\beta}$.
At the beginning of epoch $t$, if $v$ is not yet in any cluster and
$\text{start}_v=t$, $v$ becomes the cluster center of its own cluster.
During the epoch, we execute $\sr$ with failure probability $f=1/\poly(n)$,
where $S$ is the set of all clustered vertices and $R$ the set of all as-yet unclustered vertices.
Any vertex $v \in R$ receiving a message from $u \in S$ joins the cluster of $u$.
\end{description}

The algorithm {\sf Partition($\beta$)} takes
$O(\frac{\log^3 n}{\beta})$ time and $O(\frac{\log^3 n}{\beta})$ energy in $\nocd$.
Lemma~\ref{lem:property} presents some useful properties of {\sf Partition($\beta$)}.
The cluster graph is defined as the graph resulting from contracting each cluster to a vertex. 
Our strategy for solving $\broadcast$ is to iteratively apply the clustering algorithm {\sf Partition($\beta$)} to the cluster graph
until it has diameter $\poly(\log n)$.
In Lemma~\ref{lemma:7} we prove that the diameter of the cluster graph shrinks by a factor of $O(\beta)$ with high probability.

\begin{lemma}[\cite{miller2013parallel,haeupler2016faster}] \label{lem:property}
The algorithm {\sf Partition($\beta$)} partitions the vertices into clusters with the following properties.
\begin{enumerate}
\item \label{p2} The probability of any edge $\{u,v\}$ having its endpoints $u$ and $v$ contained in different clusters is at most $2\beta$.
\item \label{p3} For any fixed vertex $u$, the probability that vertices in $N^d(u) \cup \{u\}$ are in at least $t$ distinct clusters is at most $\left(1-e^{-(2d+1)\beta}\right)^{t-1}$. As a special case, for $d=1$ (i.e., if we only care about $u$ and its neighbors) this probability is at most $\left(1-e^{-3\beta}\right)^{t-1}$.
\end{enumerate}
\end{lemma}
\begin{proof}
The two properties are due to~\cite[Corollary 3.7]{haeupler2016faster} and~\cite[Corollary 3.8]{haeupler2016faster}, respectively. 
\end{proof}

\begin{lemma}[Concentration bound on diameter]\label{lemma:7}
Suppose that the diameter of the graph $G$ is  $D = \frac{\alpha \log^2 n}{\beta^4}$, for some number $\alpha$.
Then the diameter of the cluster graph resulting from {\sf Partition($\beta$)} is at most $3\beta{D}$,
with probability $1 - n^{- \Omega(\alpha)}$.
\end{lemma}
\begin{proof}
Let $k = 2\cdot \frac{2\log n}{\beta}$ be twice the number of epochs, and so the maximum diameter of any cluster is at most $k$.
Consider any two vertices $u$ and $v$ such that $\dist(u,v) > 3 \beta D = 3 \beta \cdot \frac{\alpha \log^2 n}{\beta^4} = \frac{3\alpha \log^2 n}{\beta^3}$. Let $P=(w_1, w_2, \ldots, w_\ell, w_{\ell+1})$ be a shortest path from $u = w_1$ to $v = w_{\ell+1}$ of length $\ell$. Define $X_i$ to be the indicator random variable that $w_{i}$ and $w_{i+1}$ are contained in different clusters.
Then $X = \sum_{i=1}^{\ell} X_i$ is an upper bound of the distance between the cluster of $u$ and the cluster of $v$ in the cluster graph.

If $|i-j| > k = \frac{4\log n}{\beta}$, then $X_i$ and $X_j$ are independent.
Thus, we can color $\{X_i\}_{i = 1 , \ldots, \ell}$ by $\chi = \frac{4\log n}{\beta}$ colors in such a way that
variables of the same color are independent.
By~\cite[Theorem 3.2]{DubhashiPanconesi09}, we have the following inequality:
$\Prob[X \geq \Expect[X] + t] \leq \exp(-2t^2 / (\chi \cdot \ell)).$
By linearity of expectation and Lemma~\ref{lem:property}(\ref{p2}), $\Expect[X] \le 2\beta \ell$. Thus, by setting $t = \beta \ell$, we have
\[
\Prob[X \geq 3 \beta \ell] \leq \exp(-\Omega(\beta^3 \ell / \log n)) = n^{- \Omega(\alpha)}.
\]
The lemma follows by a union bound over all $O(n^2)$ possible pairs $\{u,v\}$. Notice that if $\dist(u,v)  \leq 3\beta D$, then the distance between the cluster of $u$ and the cluster of $v$ in the cluster graph is already at most $3\beta D$.
\end{proof}

\subsection{Main Algorithm}\label{subsect:main-diam}

We fix the parameter $\beta=\frac{1}{\log^{1/\epsilon} n}$. Our randomized $\nocd$ algorithm for $\broadcast$ consists of two phases. 
The first phase is to iteratively run {\sf Partition($\beta$)} on the current cluster graph $\log_{1/(3\beta)} D$ times. 
The second phase is to apply Lemma \ref{lem:reduce} to the last clustering to solve $\broadcast$.

\paragraph{Details of the First Phase.} After performing one iteration of {\sf Partition($\beta$)} to get a new clustering, we will later see in Section~\ref{subsect:maintain} that the maximum number of layers in any cluster is multiplied by at most $\frac{4 \log n}{\beta} = 4\log^{1+\frac{1}{\epsilon}}n$. Thus, throughout the first phase, the maximum number of layers of the underlying good labeling is upper bounded by
\[
\mathcal{D} = \left( \frac{4 \log n}{\beta} \right)^{\log_{1/(3\beta)} D} 
= D^{\left(\frac{\log \frac{4 \log n}{\beta} }{\log \frac{1}{3 \beta}}\right)}
= D^{\left(\frac{\log (4\log^{1+\frac{1}{\epsilon}}n) }{\log (\frac{1}{3}\log^{\frac{1}{\epsilon}} n)}\right)}
= D^{1+\epsilon(1+ O(1/\log \log n))}.
\]
By Property~\ref{p3} of Lemma~\ref{lem:property}, with high probability, for each vertex $u$, the number of distinct clusters that vertices in $N^+(u) = N(u) \cup \{u\}$ belong to is at most 
\[
\mathcal{C} = O\left(\log_{1/3\beta} n\right) = O\left(\log_{\log^{1/\epsilon} n} n\right)= O\left(\frac{\epsilon\log n}{\log \log n}\right).
\]
We will later see that, based on the implementation of the cluster structure in Section~\ref{subsect:cluster}, we can simulate one round of {\sf Partition($\beta$)} 
on the cluster graph using $O(\mathcal{D}\mathcal{C}\log^3 n)$ rounds and $O(\mathcal{C}\log^3 n)$ energy in the underlying graph $G$. 
The details are described in Section~\ref{subsect:beep} and Section~\ref{subsect:maintain}. In Section~\ref{subsect:beep} we present a simulation of {\sf Partition($\beta$)} on the cluster graph. In Section~\ref{subsect:maintain} we show how we maintain the good labeling underlying the clustering.
Therefore, the runtime of the first phase is
$\log_{1/(3\beta)} D \cdot   O(\log^{3+1/\epsilon} n) \cdot O(\mathcal{D} \mathcal{C} \log^3 n),$
and the energy cost is
$\log_{1/(3\beta)} D \cdot   O(\log^{3+1/\epsilon} n) \cdot O(\mathcal{C}\log^3 n).$

\paragraph{Details of the Second Phase.} In view of Lemma~\ref{lemma:7}, after the first phase, the diameter of the cluster graph is less than $O(\frac{\log^2 n}{\beta^4}) = O(\log^{2+4/\epsilon}n)$. Applying Lemma \ref{lem:reduce} with $d= O(\log^{2+4/\epsilon}n)$ and $L = \mathcal{D} = D^{1+\epsilon(1+O(1/\log \log n))}$, 
$\broadcast$ can be solved in $O(D^{1+\epsilon(1+O(1/\log \log n))} \log^{4+4/\epsilon} n)$ time using $O(\log^{4+4/\epsilon} n)$ energy. 
Notice that the diameter of a cluster graph (for a specific clustering resulting from a good labeling $\mathcal{L}$) is greater than or equal to the diameter of $G_{\mathcal{L}}$.

\medskip

By doing a variable change $\epsilon' = \epsilon(1+O(1/\log \log n))$, we have the following theorem.

\begin{theorem}\label{theorem:1}
For any $\epsilon\in(0,1)$, there is a randomized $\broadcast$ algorithm in $\nocd$ taking
$O(D^{1+\epsilon} \log^{O(\frac{1}{\epsilon})} n)$ time and using
$O(\log^{O(\frac{1}{\epsilon})}n)$ energy, that succeeds with high probability.
\end{theorem}

\subsection{Cluster Structure}\label{subsect:cluster}

We assume that each vertex $v$ has a unique number $\ID(v)$, and has a good labeling $\mathcal{L}(v)$.
Recall that a good labeling, in general, does not give rise to a unique clustering.
To fix a specific clustering, consider the following modifications.
We define the cluster id of a cluster $C$ by $\ID(r)$, where $r$ is the unique layer-0 vertex in $C$.
We assume that each vertex $v \in C$ knows the cluster id $\CID(v) = \ID(r)$.
We assume the cluster center $r$ has generated a sufficiently long random string $R(r)$, and each vertex $v \in C$ knows $R(v)\bydef R(r)$.
We call this random string the {\em shared random string} of the cluster $C$.

Suppose that all vertices agree on the two parameters $\mathcal{C}$ and $\mathcal{D}$ meeting the following conditions.
For each vertex $u$, the vertices in $N^+(u)$ belong to at most $\mathcal{C}$ distinct clusters.
The number $\mathcal{D}$ is an upper bound on the number of layers of the good labeling.
We claim that the following two tasks can be done with $O(\mathcal{C}\log^3 n)$ time and $O(\mathcal{C}\log^3 n)$ energy.

\begin{itemize}
\item {\sf Downward transmission.} Let $i \geq 0$ and $V'$ be a subset of layer-$i$ vertices that have some messages to send. 
The goal is to have each layer-$(i+1)$ vertex with at least one $V'$-neighbor \emph{in the same cluster}
receive a message from any such neighbor, with high probability.

\item {\sf Upward transmission.} Let $i > 0$ and $V'$ be a subset of layer-$i$ vertices that have some messages to send. 
The goal is to have each layer-$(i-1)$ vertex with at least one $V'$-neighbor \emph{in the same cluster}
receive a message from any such neighbor, with high probability.
\end{itemize}

\begin{lemma}\label{lemma:cluster-transmit}
In the $\nocd$ model, both {\sf Downward transmission} and  {\sf Upward transmission} can be solved by a randomized algorithm that takes  $O(\mathcal{C}\log^3 n)$ time and $O(\mathcal{C}\log^3 n)$ energy.
\end{lemma}
\begin{proof}
We only present the proof for {\sf Downward transmission}, since {\sf Upward transmission} can be solved analogously.
The algorithm is as follows. Repeat the following procedure for $O(\mathcal{C}\log n)$ iterations.
Each layer-$i$ vertex $v\in V'$ joins the set $S$ with probability
$\frac{1}{\mathcal{C}}$, using the shared random string $R(v)$.
Thus, for any two layer-$i$ vertices $u,v\in V'$ in a cluster $C$,
we must have either $u,v \in S$ or $u,v \notin S$.
Run $\sr$ with $S$ being the above set, and $R$ being the set of all layer-$(i+1)$ vertices.
This algorithm takes $O(\mathcal{C}\log^3 n)$ time and $O(\mathcal{C}\log^3 n)$ energy.

Now we prove the correctness of this algorithm. Consider any layer-$(i+1)$ vertex $v$ in cluster $C$, let $u_1,\ldots,u_{x}$ be all layer-$i$ neighbors of $v$ in $C$ that are transmitting, and let $u_{x+1},\ldots,u_k$ be all layer-$i$ neighbors of $v$ not in $C$ that are transmitting. The vertices $u_1,\ldots,u_k$ are contained in at most $\mathcal{C}$ distinct clusters. Within $O(\mathcal{C}\log n)$ iterations, with high probability, there is an iteration  where (i) $u_1,\ldots,u_{x} \in S$, and (ii) $u_{x+1},\ldots,u_k \notin S$. Thus, $v$ is able to receive a message from a neighbor in $C$ in this iteration. We assume any message contains the cluster id, so that $v$ can check whether a message it receives comes from a neighbor in $\mathcal{C}$.
\end{proof}

\subsection{Simulating Algorithms on Cluster Graph}\label{subsect:beep}
In view of the definition of $\sr$, we define the $\cd^\star$  model as follows. This model is basically the same as $\cd$; but for the case where at least two neighbors are transmitting, the listener receives any one of these messages (instead of receiving noise). The choice of the message that the listener receives can be arbitrary. Observe that {\sf Partition($\beta$)} works in $\cd^\star$.

Notice that $\cd^\star$ is strictly stronger than the {\em beeping model}~\cite{cornejo2010deploying}, 
which is defined as follows. In each round a vertex can either beep, listen or remain idle. Beeping and idle vertices receive no feedback, and listening vertices can differentiate between (i) the case where at least one of its neighbors are beeping, and (ii) the case where none of its neighbors are beeping. 

\paragraph{Simulation.} Consider one round of $\cd^\star$ on the {\em cluster graph} (the graph resulting from contracting each cluster into a vertex).
Let $\mathcal{S}$ be the set of all clusters that are transmitting, and let $\mathcal{R}$ be the set of all clusters that are listening.
This round can be simulated in the underlying graph $G$ by the following three operations: (i) {\sf Down-cast} allows the center of each cluster $C \in \mathcal{S}$ to broadcast a message to the entire cluster; (ii)  {\sf All-cast} allows messages to be transmitted between the clusters; (iii) {\sf Up-cast} allows the center of each cluster $C \in \mathcal{R}$ to obtain one message sent to the cluster (if there is any). Recall that $\mathcal{D}$ is an upper bound for the number of layers.
\begin{itemize}
\item {\sf Down-cast.} For each $C \in \mathcal{S}$, the center $r$ of $C$ generates some message $m$, and the goal is to let all vertices in $C$ know $m$.
This can be done by transmitting the message layer by layer. The algorithm is as follows. For $i = 0, \ldots, \mathcal{D}-2$, suppose all layer-$i$ vertices have received the message, and then execute {\sf Downward transmission} to let all layer-$(i+1)$ vertices to receive the message. This operation requires $O(\mathcal{D}\mathcal{C}\log^3 n)$ time and $O(\mathcal{C}\log^3 n)$ energy.

\item {\sf All-cast.} Let $S$ be the set of all vertices that belong to a cluster in $\mathcal{S}$, and let $R$ be the set of all vertices that belong to a cluster in $\mathcal{R}$. Each $v \in S$ has a message to transmit, and the goal is to let each $u \in R$ such that $N(u) \cap S \neq \emptyset$ to receive some message.
    This can be solved in a way analogous to Lemma~\ref{lemma:cluster-transmit}. This operation requires $O(\mathcal{C}\log^3 n)$ time and $O(\mathcal{C}\log^3 n)$ energy.

\item {\sf Up-cast.} For each $C \in \mathcal{R}$, some vertices in a cluster $C$ hold a message, and the goal is to let the center know any one of them, if at least one exists.
The algorithm is similar to {\sf Down-cast}. For $i = \mathcal{D}-1, \ldots, 1$, run {\sf Upward transmission} to let layer-$(i-1)$ vertices to receive messages from layer-$i$ vertices. This operation requires $O(\mathcal{D}\mathcal{C}\log^3 n)$ time and $O(\mathcal{C}\log^3 n)$ energy.
\end{itemize}

\begin{lemma}\label{lemma:simul-beeping}
In the $\nocd$ model, we can simulate any $\cd^\star$  algorithm on  the cluster graph, where each round of the $\cd^\star$  
algorithm is simulated in $O(\mathcal{D}\mathcal{C}\log^3 n)$ time using $O(\mathcal{C}\log^3 n)$ energy.
\end{lemma}

\subsection{Maintaining Good Labeling} \label{subsect:maintain}

In this part, we show the details of how we can maintain the good labeling $\mathcal{L}$ as well as other information, 
such as the cluster id $\CID(v)$ and shared random string $R(v)$, 
while some clusters are being merged. 

\medskip

Let $W$ denote the set of all vertices that successfully received ``merging requests'' at some time during an execution of {\sf Partition($\beta$)} (more precisely, at an {\sf All-cast} operation in Section~\ref{subsect:beep}).
We assume that the merging request sent from a vertex $v \in C'$ contains the following information: $\ID(v)$, $\CID(v)$, $R(v)$, and $\mathcal{L}(v)$.  For each $u \in W$, let $\phi(u)$ be the vertex in a neighboring cluster that successfully sent the merging request to $u$.

Each cluster $C$ with $C \cap W \neq \emptyset$ needs to select one vertex $v^\star \in C \cap W$,  re-root the cluster $C$ at $v^\star$, 
and assign a new good labeling $\mathcal{L}'$ to all $C$-nodes.
This can be done by applying {\sf Up-cast} and {\sf Down-cast} (in Section~\ref{subsect:beep}) 
on the old labeling $\mathcal{L}$. 
That is, this task can be accomplished in $O(\mathcal{D}\mathcal{C}\log^3 n)$ time 
using $O(\mathcal{C}\log^3 n)$ energy.  The algorithm is as follows.

\paragraph{Step 1: Electing $v^\star$.} Perform an {\sf Up-cast} to let the cluster center of $C$ elect a vertex $v^\star \in C \cap W$, and then perform a {\sf Down-cast} to let all vertices in $C$ know.

\paragraph{Step 2: Update Labeling $\mathcal{L}'$.} Initially, all vertices $v \in C$ have  $\mathcal{L}'(v) = \bot$, except that $\mathcal{L}'(v^\star)$ is initialized as the layer of $\phi(v^\star)$ plus 1. The $\mathcal{L}'$-label of all vertices in $C$ can be computed by {\sf Up-cast} and {\sf Down-cast} as follows.
\begin{itemize}
\item Perform an {\sf Up-cast}. The message of $v^\star$ is its $\mathcal{L}'$-label. Each vertex $v$ that receives a message $m$ sets $\mathcal{L}'(v) = m+1$, and it will transmit the message $m+1$ during the next {\sf Upward transmission}.
\item Perform a {\sf Down-cast}. For each vertex $v$ that has obtained a $\mathcal{L}'$-label, its message is its $\mathcal{L}'$-label (and it will not reset its $\mathcal{L}'$-label). Each vertex $v$ that has not obtained a $\mathcal{L}'$-label sets $\mathcal{L}'(v) = m+1$, where $m$ is the message it receives.
\end{itemize}

Notice that information about cluster id and shared random string can also be transmitted through the above procedure (Step~2).

\newcommand{\ind}{\operatorname{Ind}}

\newcommand{\Active}{\mathsf{Active}}
\newcommand{\Wait}{\mathsf{Wait}}
\newcommand{\Halt}{\mathsf{Halt}}

\section{Improved Randomized Algorithms for the $\cd$ Model}\label{sect:improve-cd}

In this section we show that the energy complexity for the randomized $\cd$ model in Section~\ref{section:rand} can be further improved to nearly match the lower 
bound, up to a small $O(\log \log \Delta / \log \log \log \Delta)$ factor.  The price for this energy efficiency is a super-linear running time.
The key idea is to exploit the following properties of Lemma~\ref{lemma:1-mod}: (i) if each $v \in S$ is adjacent to at most one vertex in $R$, then $\sr$ consumes $O(\log \log \Delta)$ energy {\em in expectation}; and (ii) if $f$ is {\em sufficiently large} (e.g., $f = 1/\log \Delta$), then $\sr$ also consumes $O(\log \log \Delta)$ energy.

The main algorithm of this section still follows the high-level structure of the algorithm in Section~\ref{section:rand}.
That is, we begin with the trivial all-0 good labeling, and then repeat the procedure of ``computing a new labeling $\mathcal{L}'$ from old labeling $\mathcal{L}$'' for several iterations.
In the algorithm in Section~\ref{section:rand}, we do not explicitly maintain a fixed cluster structure, and it is possible that a ``cluster'' is eaten by multiple adjacent 
``clusters'' during {\sf Up-cast}. Here we use a certain vertex coloring to fix the clustering and a spanning tree of each cluster.
We will only allow a cluster to be merged, in its entirety, into exactly one adjacent cluster. 
The vertex coloring is described in Section~\ref{subsect:coloring}.
The procedure of merging clusters is described in Section~\ref{subsect:merge}.

\subsection{Vertex Colorings}\label{subsect:coloring}
Let $c\ge 1$ and $\xi > 0$ be two parameters to be determined.
Consider $c$ random $n^{\xi} \Delta$-coloring of vertices.
We write $\Color_i(v)$ to denote the color of $v$ in the $i$th coloring.
We write $\ID(v)$ to denote $(\Color_1(v), \ldots, \Color_c(v))$.
For each ordered pair of neighboring vertices $(u,v)$, we write $\ind(u,v)$ to denote the smallest index $i$ such that $\Color_i(v)$ is different from $\Color_i(w)$ for all $w \in N(u) \setminus \{v\}$. The probability that $\ind(u,v)$ does not exist is at most $n^{-c\xi}$. By a union bound, the probability that $\ind(u,v)$ is well-defined for all ordered pairs of neighboring vertices $(u,v)$ is at least $1 - n^{2-c\xi}$. We select $c = O(1 / \xi)$ to be large enough such that $n^{2-c\xi} = 1/\poly(n)$ is negligible.

\paragraph{Cluster Structure.} Recall that a good labeling represents a clustering of the graph, and each cluster $C$ is a rooted tree $T$, where the root $r$ is the unique layer-0 vertex in the cluster $C$. In what follows, we devise an implementation of such a cluster structure, which enables a more energy-efficient implementation of the primitives 
{\sf Upward transmission} and {\sf Downward transmission}. 
For each $i>0$, we assume that each layer-$i$ vertex $u$ has a designated layer-$(i-1)$ parent $v \in N(u)$, and $u$ knows $\ID(v)$.

\begin{lemma}\label{lem:ind}
Consider the task whose goal is to have each $u$ learn $\ind(u,v)$, where $v$ is the parent of $u$.
In both the $\cd$ and $\nocd$ models, 
there is a deterministic algorithm that takes $O(n^{\xi} \Delta / \xi)$ time and $O(1/ \xi)$ energy for this task.
\end{lemma}
\begin{proof}
For $j = 1$ to $c$, and for $k = 1$ to $n^{\xi} \Delta$, do the following.
Each vertex $v'$ with $\Color_j(v') = k$ speaks; each vertex $u$ whose parent $v$ has $\Color_j(v) = k$ listens.
Then $\ind(u,v)$ is the smallest index $j$ such that $u$ successfully hears a message.
\end{proof}

In what follows, we assume that each $u$ already knows $\ind(u,\text{parent}(u))$.
We show how to use this information to efficiently relaying messages within a cluster.
Consider the following two tasks.

\begin{itemize}
\item {\sf Upward transmission.} Let $i > 0$ and $V'$ be a subset of layer-$i$ vertices that have some messages to send. The goal is to have each layer-$(i-1)$ vertex $v$ that has a child in $V'$ to receive some message from one of its child (with high probability). This task can be done in time $O(\log \Delta \log n \cdot (n^{\xi} \Delta / \xi))$ with energy cost $O(\log \log \Delta + (1/\xi)) + X$, where $X \sim \operatorname{Exponential}(\lambda)$, for some $\lambda = O(1)$. This is done by applying Lemma~\ref{lemma:1-mod}, as follows.

    For $j = 1$ to $c$, and for $k = 1$ to $n^{\xi} \Delta$, do $\sr$: $S$ is the set of all vertices $u$ in $V'$ such that (i) its parent $v$ has $\Color_j(v) = k$, and (ii) $\ind(u,v) = j$;  $R$ is the set of all layer-$(i-1)$ vertices $v'$  with  $\Color_j(v') = k$ that have yet to receive a message from a child.\footnote{Suppose that $v$ is the parent of  $u \in V'$, and we have   $\ind(u,v) = j$ and $\Color_j(v) = k$. For any $w \in N(v) \cap V'$ that is not a child of $v$, we cannot simultaneously have $\ind(w,\text{parent}(w)) = j$ and $\Color_j(\text{parent}(w)) = k$ due to the definition of $\ind(\cdot,\cdot)$ (since $v \in N(w)$ and $\Color_j(v) = k$). Thus, our choice of $S$ and $R$ ensures that the communication occurs only between a parent and its children.} Notice that in one round we can let each vertex $v' \in R$ check whether it has a neighbor in $S$; and if there is no such neighbor, then $v'$ can skip this $\sr$. Thus, the energy cost for each vertex is $c = O(1/\xi)$ plus the energy cost for one $\sr$.

\item {\sf Downward transmission.} Let $i \geq 0$ and $V'$ be a subset of layer-$i$ vertices that have some messages to send. The goal is to have each $v \in V'$ deliver its message to \emph{all} its children, with zero probability of failure. This task can be done in time $O(n^{\xi} \Delta / \xi)$ with energy cost $O( 1 / \xi)$, as follows.

    For $j = 1$ to $c$, and for $k = 1$ to $n^{\xi} \Delta$, do the following. Each vertex $v' \in V'$ with $\Color_j(v') = k$ sends its message; each layer-$(i+1)$ vertex $u$ listens to the channel if (i) its parent $v$ has $\Color_j(v) = k$, and (ii) $\ind(u,v) = j$.
\end{itemize}

In the $\cd$ model, a vertex is said to be {\em relevant} to an {\sf Upward transmission} task if either (i) $v \in V'$, or (ii) $v$ has a child in $V'$. Notice that a vertex irrelevant to the task does not need to accomplish anything.
We can lower the energy cost of irrelevant layer-$(i-1)$ vertices to just $O( 1 / \xi)$ (from $O(\log \log \Delta + (1/\xi)) + X$).
This is due to the observation that in $O( 1 / \xi)$ energy and $O(n^{\xi} \Delta / \xi)$ time, each layer-$(i-1)$ vertex can know whether it is relevant to the current  {\sf Upward transmission}  task (i.e., whether it has a child in $V'$).

\subsection{Algorithm for Merging Clusters}\label{subsect:merge}

Let $p,f \in(0,1)$ and $s\geq1$ be three parameters to be chosen.
Given a good labeling $\mathcal{L}$ and a clustering consistent with $\mathcal{L}$, 
the goal is to obtain a new good labeling $\mathcal{L}'$ and clustering with fewer clusters.
Consider the following procedure on the {\em cluster graph}. During the execution of the procedure, a cluster $C$ is in one of three possible states: 
$\Active$, $\Wait$, and $\Halt$.
\begin{enumerate}
\item  Each cluster $C$ initiates a group, which consists of only one member, namely $C$.
\item  \label{step:xx} Initially, each cluster $C$ is $\Active$ with probability $p$, and the remaining ones are $\Wait$.
\item  \label{step:x} Repeat for $s$ iterations:
\begin{enumerate}
\item Each $\Active$ cluster $C$ broadcasts a ``merging request'' to all its neighboring $\Wait$ clusters, and then $C$ resets its status to $\Halt$.
\item For each $\Wait$ cluster $C$ that has some neighbors sending merging requests, with probability at least $1-f$, $C$ successfully receives one request $\gamma$, the group of $C$ is merged into the group of the cluster that sends the request $\gamma$, and then $C$ resets its status to $\Active$.
\end{enumerate}
\item The groups of clusters becomes the new clustering.
\end{enumerate}
In what follows, we show how to implement the above procedure in the underlying graph $G$ such that by the end we obtain a new good labeling $\mathcal{L}'$ that represents the new clustering. We let the root of the cluster be the one who makes all decisions for the cluster. We can use $n-1$ iterations of {\sf Upward transmission} to gather information (e.g., merging requests) from other vertices to the root, and use $n-1$ iterations of {\sf Downward transmission} to broadcast information (e.g., status of the cluster) from the root to other vertices. At the beginning, we first execute the algorithm of Lemma~\ref{lem:ind} to have each $u$ learn $\ind(u,\text{parent}(u))$.

\paragraph{Initializing the Labeling $\mathcal{L}'$.} For each cluster $C$ that is $\Active$ at Step~\ref{step:xx}, all vertices $v \in C$ initialize $\mathcal{L}'(v) = \mathcal{L}(v)$. This can be done using $n-1$ iterations of {\sf Downward transmission}. All remaining vertices $u$ initialize $\mathcal{L}'(u) = \bot$.

\paragraph{Merging Requests.} Suppose that each vertex in a cluster $C$ knows the status of $C$.
We use Lemma~\ref{lemma:1-mod} to implement the transmission of merging requests. Do $\sr$ (with success probability $1-f$) with $S$ being the set of all vertices that are in an $\Active$ cluster, $R$ being the set of all vertices that are in an $\Wait$ cluster. The content of the request sent from $v \in S$ consists of (i) $\ID(v)$ and (ii) the layer of $v$ in $\mathcal{L}'$.

\paragraph{Maintaining a Good Labeling $\mathcal{L}'$.} We need to maintain the new good labeling $\mathcal{L}'$ and its associated parent-child 
relation while simulating Step~\ref{step:x} on the underlying graph $G$.
Consider the $i$th iteration of Step~\ref{step:x}.
Let $W$ denote the set of all vertices that successfully received merging requests in this iteration.
For each $\Wait$ cluster $C$ such that $C \cap W \neq \emptyset$, it needs to select one vertex $v^\star \in C \cap W$, and re-root the tree $T$ of the cluster $C$ at $v^\star$.
This can be done using the algorithm of Section~\ref{subsect:maintain}, and it takes $O(n)$ calls to 
{\sf Upward transmission} and {\sf Downward transmission}, but each vertex in $C$ participates in $O(1)$ of them.

Notice that for a vertex $v$ in a $\Wait$ cluster $C$ such that $C \cap W = \emptyset$, we still run the algorithm of Section~\ref{subsect:maintain}, but no merging occurs. In this case, $v$ is irrelevant to all {\sf Upward transmission} used in the algorithm of Section~\ref{subsect:maintain}, since no one is attempted to transmit a message in $C$.

\ignore{
For each $u \in W$, let $\phi(u)$ be the vertex in a neighboring cluster that successfully sent the merging request to $u$.  For each $u \in C \setminus X$, let $\phi(u) = \bot$ be initially undefined. Intuitively, $\phi(u)$ represents a potential parent of $u$ in the new labeling $\mathcal{L}'$. For each $u \in W$, we initially set $\mathcal{L}'(u) = \mathcal{L}'(\phi(u))+1$. The values $\phi(u)$ and $\mathcal{L}'(u)$ for each $u \in C$ may be reset later on. We do {\sf Up-cast} and then {\sf Down-cast} to re-root the tree $T$.
\begin{itemize}
\item {\sf Up-cast.} For $i = n-2, \ldots, 0$, do the following task for each layer-$i$ vertex $u \in C$. If we already have $u\in W$, then nothing needs to be done for $u$. Otherwise, $u$ has to figure out if it has a child $w$ such that $\phi(w)$ is defined. If such a child $w$ exists, then $u$ sets $\phi(u)=w$ and $\mathcal{L}'(u) = \mathcal{L}'(w)+1$. This step can be implemented by $n-1$ {\sf Upward transmission}.
\item {\sf Down-cast.} For $i = 0, \ldots, n-2$, do the following task for each layer-$i$ vertex $u \in C$. The vertex $u$ broadcasts $\phi(u)$ and $\mathcal{L}'(u)$ to all its children. Each child $w$ of $u$ does the following updates.
    For the case of $\phi(u) = w$, the vertex $w$ sets its new parent to be $\phi(w)$.
    For the case of $\phi(u) \neq w$, the parent of $w$ remains to be $u$, and then $w$ (re-)sets $\phi(w)=\bot$ and $\mathcal{L}'(w) = \mathcal{L}'(u)+1$. This step can be implemented by $n-1$ {\sf Downward transmission}.
\end{itemize}
}

\paragraph{Finalizing the Labeling $\mathcal{L}'$.} After all $s$ iterations of Step~\ref{step:x}, all clusters except the ones that have not  participated in any merging operation have obtained valid $\mathcal{L}'$-labeling. For each vertex $v$ in the remaining clusters (i.e., those in state $\Wait$), 
we set $\mathcal{L}'(v) = \mathcal{L}(v)$. This gives us a desired good labeling $\mathcal{L}'$.

\paragraph{Analysis.}
We analyze the time and energy complexities in the $\cd$ model.
We make $O(ns)$ calls to {\sf Upward transmission} and {\sf Downward transmission}, but each vertex participates in only $O(s)$ of them. 
Moreover, each vertex is only {\em relevant} to $O(1)$ calls to {\sf Upward transmission}.
Therefore, the total time for {\sf Upward transmission} and {\sf Downward transmission} is $O(ns \log \Delta \log n \cdot (n^{\xi} \Delta / \xi))$, and the total energy for  {\sf Upward transmission} and {\sf Downward transmission} is $O((s/\xi) + \log \log \Delta) + Y$, where $Y$ is a summation of $O(s)$ variables drawn i.i.d. from 
$\operatorname{Exponential}(\lambda)$, for some $\lambda = O(1)$.
The transmission of merging requests is invoked $O(s)$ times. 
By Lemma~\ref{lemma:1-mod}, each of these transmissions takes energy 
$O(\log \log \Delta + \log 1/f)$ and time $O(\log \Delta (\log \log \Delta + \log 1/f))$.
To summarize, the total runtime is
$$O(ns \log \Delta \log n \cdot (n^{\xi} \Delta / \xi) + s \log \Delta (\log \log \Delta + \log\fr{1}{f})),$$
and the total energy is
$$O(s \cdot (\fr{1}{\xi} + \log \fr{1}{f}) + \log \log \Delta) + Y.$$

\subsection{Main Theorem}

Let $w$ denote the number of layer-0 vertices in $\mathcal{L}$.
Observe that each layer-0 vertex in $\mathcal{L}$ remains layer-0 in $\mathcal{L}'$ with probability at most $p' \bydef p+(1-p)^{\min\{s+1,w\}}+ f \cdot \min\{s,w-1\}$ during the procedure of merging clusters. The term $f \cdot \min\{s,w-1\}$ is an upper bound for the probability that at least one transmission of merging request fails among (at most) $\min\{s,w-1\}$ trials. 



\begin{theorem}
There is an $O\left(\frac{\log n (\log \log \Delta + (1/\xi))}{ \log \log \log \Delta}\right)$-energy and $O\left({n^{1+\xi} \Delta}\right)$-time randomized algorithm that solves $\broadcast$ with high probability in the $\cd$ model, for any $\xi = \omega(\log \log n / \log n)$.
\end{theorem}

\begin{proof}
Use the parameters $p=\log^{-1/2} \log \Delta, s=\log \log \Delta, f=\log^{-3/2} \log \Delta$.
We have $p' =O(\log^{-1/2} \log \Delta)$. Thus, after $O(\log n / \log \log \log \Delta)$ iterations, we obtain a good labeling $\mathcal{L}^\ast$ where the number of layer-0 vertices is at most $\log \log \Delta$, with high probability.  After that, we use Lemma~\ref{lem:reduce} to solve $\broadcast$.

The runtime is $O(\log n /  \log \log \log \Delta) \cdot O(ns \log \Delta \log n \cdot (n^{\xi} \Delta / \xi) + s \log \Delta (\log \log \Delta + \log 1/f))
 = O(\Delta n^{1+\xi}\cdot \frac{\log \Delta \log^2 n \log \log \Delta}{ \xi   \log \log \log \Delta })$.
The energy is $O(\frac{\log n (\log \log \Delta + (1/\xi))}{ \log \log \log \Delta}) + Z$, where $Z$ is a summation of $O( \log n /  \log \log \log \Delta)$ variables drawn i.i.d.~from
$\operatorname{Exponential}(\lambda)$, for some $\lambda = O(1)$. 
With high probability, $Z = O(\log n)$. 
Thus, the energy cost is $O(\frac{\log n (\log \log \Delta + (1/\xi))}{ \log \log \log \Delta})$.

The theorem follows from a change of variable (from $\xi$ to $\xi'$) such that $n^{\xi'} = n^{\xi} \cdot \frac{\log \Delta \log^2 n \log \log \Delta}{ \xi   \log \log \log \Delta }$. 
Notice that $\xi' = \xi + O(\log \log n / \log n)$, and so $1/\xi' = \Theta(1/\xi)$.
\end{proof}

Setting $\xi = O(1 / \log \log \Delta)$, we obtain a randomized algorithm using $O\left(\frac{\log n \log \log \Delta}{\log \log \log \Delta}\right)$ energy and 
$O\left(\Delta n^{1+O(1/\log \log \Delta)}\right)$ time, which nearly matches the $\Omega(\log n)$ energy lower bound.

\algnewcommand{\LineComment}[1]{\Statex  #1}

\section{Algorithm for the Path} \label{section:path}

In this section, we examine the special case where the underlying graph
is a path of $n$ vertices.  In this case, we will see that a broadcast can be
achieved with only a constant factor overhead in time, and only $O(\log n)$
energy per vertex.  It is clear that each of these bounds is within a
constant factor of the best possible, at least if we consider the
energy usage of the most unlucky vertex, in either collision model.
This provides some hope that our results for general
graphs can also be improved further, until they match the known lower
bounds, which we conjecture to be tight.


In light of Theorem~\ref{thm:simlocal} we
will assume we are working in the full duplex $\LOCAL$ model.  Notice that, since the path  has
maximum degree 2, the overhead for this reduction is only a constant factor.  Thus all
of our results will also hold in the $\cd$ and $\nocd$ models as well.

\subsection{The Protocol}

Pseudocode for the path algorithm is given in Algorithm~\ref{alg:path}.
The algorithm is stated for the model in which each vertex knows which
of its neighbors is ``upstream'' (closer to the message source), and
which is ``downstream.''  In the general algorithm for which this
information is not known, each vertex executes this algorithm twice, in
parallel, once with each of its two neighbors in the upstream and downstream role.
Since we are working in the $\LOCAL$ model, there is essentially no
overhead for doing this, except that the energy cost will be double.

It will be convenient to assume that $n$ is a power of 2. This is without loss of generality, since if $n$ is not a power of two, the vertices can simulate the algorithm corresponding to the path of $2^{\lceil \log_2 n\rceil}$ vertices, with at most a constant factor increase in time and energy.

In the algorithm, each vertex $v$ independently samples a ``blocking time'' $B_v$ from the following distribution
\[
B_v = \begin{cases}   2^b & \mbox{ with probability } 2^{-b}    \mbox{ for }  1\le b < \log_2 n \\
 n & \mbox{ otherwise }
\end{cases}
\]
After transmitting the value of $B_v$ downstream and receiving $B_w$ from its upstream neighbor in step 1, $v$ turns its transmitter off until time $B_v$. Any  synchronization messages it may receive during this time are blocked
from being transmitted further downstream. If the payload message arrives during this time, it is delayed from being transmitted further until time $B_v$. Notice that this does not mean that $v$ is listening to the traffic for $B_v$ steps. Rather, it inductively schedules the sequence of listening times  based on the messages it has heard in previous listening times, the first of these being the $B_w$ it heard on step 1.
Starting at time $B_v$,
$v$  switches gears, and begins immediately forwarding all
messages it receives.  The one special case occurs at time $B_v$:
at this time, if $v$ has already received the payload message, it
forwards it at time $B_v$; otherwise, $v$ calculates the time
remaining until it is due to forward its next message from its
upstream neighbor, and sends a message telling its downstream
neighbor how long it will need to wait.  Notice that $v$'s most recent
message from its upstream neighbor will always contain the
information needed to perform this calculation.

Intuitively, vertices with a large blocking time, $B_v$, are going to
protect the downstream vertices from the synchronization traffic
being sent by upstream vertices.  On the other hand, the larger $B_v$
is, the more of an artificial delay can be caused to the eventual
delivery of the payload message.

Algorithm~\ref{alg:path} describes the algorithm in more detail.
Figure~\ref{fig:path} illustrates a timeline of message traffic in one
direction along the path.
Our results for this algorithm are summarized in the following
theorem.

\begin{theorem}\label{thm:path}
Algorithm~\ref{alg:path} solves $\broadcast$ on an $n$-vertex path 
with worst-case running time $2n$ and expected per-vertex energy
cost $O(\log n)$.
\end{theorem}

We observe that the above upper bound on expected energy cost is within
a constant factor of the worst-vertex energy cost lower bound of
Theorem~\ref{thm:lb-path}.

\begin{algorithm}[H]
	\caption{More Efficient Broadcast on a Path}
	\label{alg:path}
	\begin{algorithmic}[1]
		\LineComment{{\bf Initialize:}}
		\State \hspace{0.5cm} If this vertex is the message source, send the payload message at
  timestep 1, then quit.
\State \hspace{0.5cm} Sample $b$ from a geometric distribution with
mean $2$.  $\Prob(b = i) = 2^{-i}$, for $i \ge 1$.
\State \hspace{0.5cm} If $2^b > n$, set $b= \log_2(n) $ instead.  
\State \hspace{0.5cm} $B \gets 2^b$.  \Comment{$B$ is the ``Blocking time''}
\LineComment{ {\bf Time step 1:}}
 \State{\hspace{0.5cm} Send message ``Next message after $B-1$ timesteps''} at timestep 1.
 \State{\hspace{0.5cm} Set SendAlarm to ring at time $B$.}

	\While{haven't quit}
		     \If{($t = 1$) or ListenAlarm rang at time $t$}
     \State{Receive and store incoming message $M$ at time $t$}
     \If{$M$ is of the form ``Next message after $i$ timesteps''}
         \State{Set ListenAlarm to ring at time $t+i$.}
     \EndIf
     \If{$t \ge B$}   \Comment{``Forwarding mode''}
          \State{Send the most recently received message $M$ downstream at time $t+1$.}
      \If{$M$ was the payload message}
       \State{quit}
       \EndIf
       \EndIf
     \EndIf
      \If{SendAlarm rang at time $t$}   \Comment{$t=B$}
        \If{the payload message was received before time $B$}
          \State Send the payload message at time $B$, then quit.
         \Else
	    \State $A \gets $ time the next ListenAlarm is set for.
           \State{Send message ``Next message after $A + 1 - B$ timesteps'' at
           time $B$.}
         \EndIf
         \EndIf
       \State{Sleep until next alarm.}	
		\EndWhile

	\end{algorithmic}
\end{algorithm}

\begin{figure}
\begin{center}
  \includegraphics[width=4in]{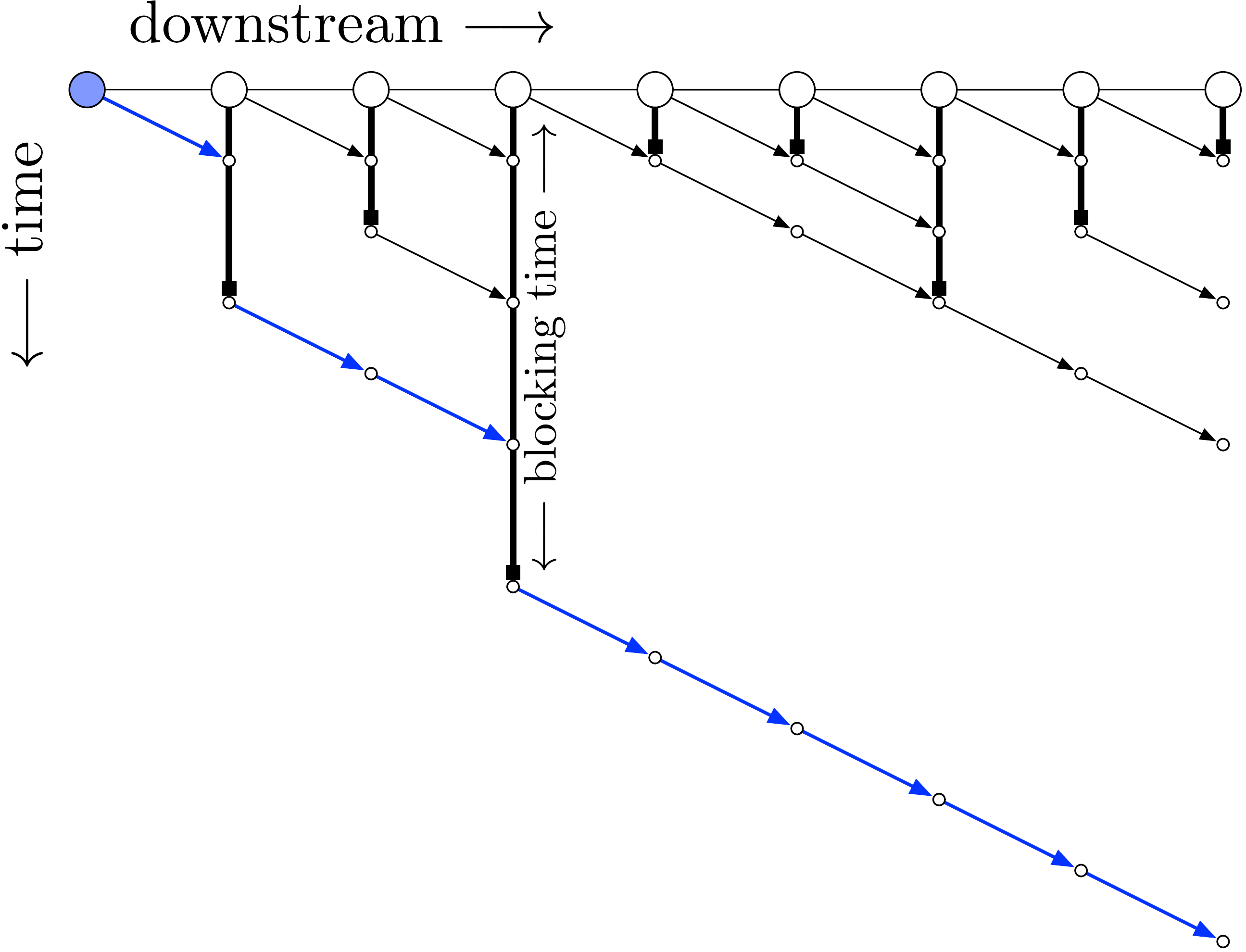}
  \caption{An example run of our path algorithm.  Messages propagate down
and to the right, except when they arrive at a blocking vertex.}
  \label{fig:path}
\end{center}
\end{figure}

%

\subsection{Analysis}

\subsubsection{Time Overhead}
We first observe that no message can be
delayed more than $\max B$ timesteps.
This is because, after time $\max B$, all vertices have gone into
``forwarding mode,'' so no further delays can occur.
Since the distribution of blocking times is supported on
powers of $2$ less than or equal to $n$,
it follows that the maximum possible delay is $n$.
Since the length of the path is $n-1$,
the algorithm has a worst-case running time of $2n-1$.


\par\noindent {\bf Remark: }
In the setting where $n$ is unknown, we could try to run the same
algorithm, but with no \emph{a priori} upper bound on the blocking times, $B$.
Then Markov's inequality, together with a union bound,
implies that the probability that $\max B > \frac{n}{\epsilon}$ is at
most $\epsilon$, for every $\epsilon > 0$.  On the other hand, the
expected value of $\max B$ is infinite, as indeed is the expected
value of even one blocking time $B$.  Thus, we have the unusual
situation that the running time is probably $O(n)$, but its expected
value is infinite.

\subsubsection{Energy Cost}
We make the following observations
\begin{itemize}
\item A message advances one step on the path per time step, except when it is blocked.
\item Each vertex $v$ originates at most two messages, one at time step 1 and another at time step $B_v$.
\end{itemize}
Thus the energy cost is controlled by the
number of incoming messages to the vertex.
Since the messages sent at time $1$ are all received and
blocked, we can ignore these.

To analyze the energy use, we compute, for each time $T$ with $2 \le T < 2n$, an upper
bound on the probability that a message reaches a particular vertex $v$ at time $T$.

First we notice the following about the distribution of the blocking times.
For $1\le b < \log n$ we have $\Prob[B=2^b] = 2^{-b}$, and $\Prob[B>2^b] = 2^{-b}$, and for $b=\log n$, $\Prob[B=2^b] = 2^{-b+1}$, and $\Prob[B>2^b] = 0$.

Let $s = \lfloor \log_2 T\rfloor$ and $t= T-2^s$ so that $T=2^s +t$. Notice that $s \le \log n$, since $T\le 2n-1$.

Consider a message that arrives at vertex $v$ at time $T$. Such a message must have originated at some time $2^i$, $1\le i \le s$ and traveled $T-2^i$ steps on the path without being blocked.

\begin{lemma} For each $i \le s$,
the probability that a message originating at time $2^i$ is not blocked before time $T$ is at most $e^{i-s}$.
\end{lemma}
\begin{proof}
Let $i \le j < s$, and consider a message that has survived to time
$2^j$, either because $j=i$ and it has just originated (base case), or
because $j>i$ and it has not been blocked so far (inductive step).
In order to survive, unblocked, for the next $2^j$ time steps,
\emph{i.e.} until time $2^{j+1}$, each of the next $2^j$ vertices must
pick a blocking time $B \le 2^j$.  Since $j < s \le \log_2 n$, this
probability equals $(1 - 2^{-j})^{2^j}$, which is at most $1/e$,
since $1 - x \le e^{-x}$ for all $x$.

By induction, we find that the probability of the message surviving
unblocked from time $2^i$ until time $2^s$ is at most $e^{i-s}$,
and of course this is an upper bound on the probability of going
unblocked until time $T$.
\end{proof}

\begin{lemma}
Fix a vertex $v$ and time $T \ge 2$.
The probability that $v$ receives a message at time $T$ is at most $\frac{4e}{(e-2)T}$.
\end{lemma}
\begin{proof}
As mentioned before, for a message to arrive at vertex $v$ at time $T$, it must have
originated at time $2^i$ for some $i$ with $1\le i \le s =\lfloor \log T\rfloor$. Moreover, it must have originated at a vertex $w$ at distance exactly $T-2^i$ from $v$, and it must not have been blocked between and $v$. This means $w$ must have selected blocking time $2^i$. This, together with the previous lemma tells us that the probabilty $p(T)$ of a message being received at $v$ at time $T$ is given by
\begin{align*}
p(T) &= \sum_{i=1}^{s} \Prob[B_w = 2^i] \cdot \zero{\Prob[w\mbox{'s message not blocked before time } T]}\\
&\le \sum_{i=1}^{s} 2^{-i+1} e^{i-s}\\
&= 2^{1-s} \cdot \sum_{i=1}^{s}  (2/e)^{s-i}\\
&\le 2^{1-s} \frac{1}{1 - (2/e)} & \mbox{summing the infinite
  arithmetic series} \\
&< \frac{4e}{(e-2)T}  & \mbox{ since $T < 2^{s+1}$.}
\end{align*}
as claimed.
\end{proof}


%

%

Summing this for $T \in \{2, \dots, 2n-1\}$, and approximating the
resulting harmonic series by a logarithm, we find that the total
expected number of messages received at a given vertex is at most $\left(\frac{4e}{e-2}\right)
\ln(2n-1)$.  This establishes that the per-vertex expected energy cost is $O(\log
n)$.

\section{Conclusion}\label{sect:conclusion}

\emph{Energy complexity} is a natural and attractive concept in wireless radio networks. 
In this work we presented what we believe are the first theoretical results on the energy complexity
of problems in multi-hop networks.  It is interesting that many of the techniques we used (lots of sleeping,
tightly scheduled transceiver usage, 2-hop neighborhood coloring) 
are somewhat similar to techniques suggested in systems papers~\cite{vanDam2003,Ye2004,Wang2006,Hong2009,Hong13}, 
but without rigorous asymptotic guarantees.

There are several difficult problems left open by this work.  Assuming energy usage is paramount, is it possible
to design $\broadcast$ algorithms meeting our best lower bounds: $\Omega(\log n)$ in $\cd$ and $\Omega(\log\Delta\log n)$ in $\nocd$? 
Is it possible to get the best of both worlds: near optimality in time and energy?  Specifically, is there a small 
constant $c$ for which $O(D\log^c n)$ time and $O(\log^c n)$ energy suffice to solve $\broadcast$?
Alternatively, can our bounds for the path be extended to general graphs, getting $O(n)$ time and $O(\log n)$ per-node energy use?




\bibliographystyle{plain}
\bibliography{references}

\newpage

\appendix
\section{Deterministic Algorithms}\label{section:det}
We begin with defining $\sr$ in the deterministic model.
Let $S$ and $R$ be two (not necessarily disjoint) vertex sets.
Each vertex $u \in S$ attempts to broadcast a message $m_u$, and each vertex in $R$ attempts to receive a message.
We assume that $m_u \in \{1,2, \ldots, M\}$, and the goal of $\sr$ is to let each vertex $v \in R$ such that $N^+(v) \cap S \neq \emptyset$ knows $m_u$ for any one vertex $u \in N^+(v) \cap S$.

\begin{lemma}\label{lemma:9}
In deterministic $\cd$ model, $\sr$ can be solved in $O(\min\{M, N\})$ time and uses $O(\min\{\log M, \log N\})$ energy.
\end{lemma}
\begin{proof}
We first consider the case where $M \leq N$.
Define  $f_v \bydef \min_{u \in N^+(v) \cap S} m_u$.
The proof idea is to do a binary search to determine all $\log M$ bits of $f_v$. We write $p_x(s)$ as the length-$x$ prefix of a binary string $s$. Suppose at some moment each vertex $v \in R$ knows the first $x$ bits of $f_v$ and whether all message in $v$'s neighborhood share the same prefix $p_x(f_v)$ or not. The following procedure, which takes $O(2^x)$ time and $O(1)$ energy, let each $v \in R$ learn the $(x+1)$th bit of $f_v$ and whether all message in $v$'s neighborhood share the same prefix $p_{x+1}(f_v)$ or not. Each vertex $u \in S$ broadcasts a message at time $p_{x+1}(m_u)$ (this is interpreted as a binary number), each vertex $v \in R$ listens at time $p_x(f_v) \circ 0$ and $p_x(f_v) \circ 1$ (they are interpreted as binary numbers). Due to collision detection, $v$ can learn whether the $(x+1)$th bit of $f_v$ is $0$ or $1$. Notice that a vertex $v \in S \cup R$ does not need to simultaneously send and listen in our algorithm. The total runtime is $\sum_{x=0}^{\log M - 1} O(2^x) = O(M)$; the total energy cost is $\sum_{x=0}^{\log M - 1} O(1) = O(\log M)$.

\medskip
For the case of $M > N$, we first do the above procedure on the space $\{1, \ldots, N\}$ instead of $\{1, \ldots M\}$; and let $\ID(u)$ be the message of each $u \in S$. After this step, each $v \in R$ learns $\ID(u')$ for some $u' \in N^+(v) \cap S$. Then we allocate $O(N)$ time slots, where each $u \in S$ sends its message $m_u$ at slot $\ID(u)$. Then $v \in R$ can learn $m_{u'}$ by listening at slot $\ID(u')$.
\end{proof}

Notice that in deterministic $\LOCAL$ model, we can solve $\sr$ in $O(1)$ time and $O(1)$ energy, and each vertex in $R$ can obtain all messages sent from $N^+(v) \cap S$.

\subsection{Algorithm in $\LOCAL$ Model}
As a warm-up exercise, we present a deterministic algorithm in $\LOCAL$ model. An {\em $(\alpha,\beta)$-ruling set} of a graph $G$ is a set of vertices $I$ such that (i) for any two distinct vertices $u, v \in I$, we have $\dist(u,v) \geq \alpha$, and (ii) for any vertex $u$ in the graph, there is a vertex $v \in I$ with $\dist(u,v) \geq \beta$. A $(k,(k-1)\log N)$-ruling set can be computed in $O(k\log N)$ rounds in deterministic $\LOCAL$~\cite{awerbuch1989network} model.
Any two vertices $u$ and $v$ in a $(3,2\log N)$-ruling set $I$ satisfy $N(u)\cap N(v)=\emptyset$, and thus $|I| \leq |V| / 2$.

\paragraph{Computing a New Labeling $\mathcal{L}'$ from $\mathcal{L}$.}
Suppose that we have a good labeling $\mathcal{L}$; and the number of layer-0 vertices is $w \geq 2$.
The goal is to produce a new good labeling $\mathcal{L}'$ such that the number of layer-0 vertices is at most $w/2$.
The high level idea is to compute a $(3,2\log N)$-ruling set $I$ of $G_\mathcal{L}$,  let $I$ be the set of layer-0 vertices of $\mathcal{L}'$, and update the labeling of the remaining vertices using techniques in Section~\ref{section:rand}.

The mode detailed description of an algorithm for computing  $\mathcal{L}'$ is as follows.
The first step is to find a $(3,2\log N)$-ruling set $I$ of $G_\mathcal{L}$.
Observe that one round in $G_\mathcal{L}$ can be simulated using $O(1)$ energy and $O(n)$ time in $G$ using techniques similar to Lemma~\ref{lem:reduce}.
Thus, this step takes $O(\log N)$ energy and $O(n \log N)$ time.
The second step is to run the algorithm of computing $\mathcal{L}'$ in Section~\ref{section:rand} with $s = 2\log N$; but initialize $\mathcal{L}'(v)=0$ for all $v \in I$, and $\mathcal{L}'(v)=\bot$ for all remaining vertices. This step takes $O(\log N)$ energy and $O(n \log N)$ time. Since $I$ is a $(3,2\log N)$-ruling set of $G_\mathcal{L}$, we obtain a good labeling $\mathcal{L}'$ after this step.

\begin{theorem}
There is an $O(n\log n \log N)$-time and $O(\log n \log N)$-energy deterministic algorithm that solves $\broadcast$ in $\LOCAL$ model.
\end{theorem}
\begin{proof}
The $\broadcast$ algorithm begins with the trivial all-0 good labeling. Then we repeat for $O(\log n)$ iterations of constructing a new good labeling $\mathcal{L}'$ from the current good labeling $\mathcal{L}$ using the above algorithm. We end up with a good labeling $\mathcal{L}^\star$ with only one layer-0 vertex. Then we can run the algorithm of Lemma~\ref{lem:reduce} to solve $\broadcast$.
\end{proof}

\subsection{Algorithm in $\cd$ Model}

The basis of our algorithm is an $\cd^\star$ algorithm for $(2,\log N)$-ruling set.

\begin{lemma}\label{lemma:ruling-set}
In $\cd^\star$ model, there is a deterministic algorithm that computes a $(2,\log N)$-ruling set using $O(N)$ time, $O(\log N)$ energy, and uses only 1-bit messages.
\end{lemma}
\begin{proof}
The algorithm is essentially the same as the one in~\cite{awerbuch1989network}; but the recursive calls are done sequentially rather in parallel. The algorithm is as follows. Divide the set of vertices $V$ into two sets $V_0$ and $V_1$ according to the first bit of the ID. For $i=0,1$, recursively compute a $(2,\log N -1)$-ruling set $I_i$ of $V_i$ (using the IDs of last $\log N - 1$ bits). Then $I = I_0 \cup \{v\in I_1  \ | \ N(v)\cap I_0 = \emptyset\}$ is a $(2, \log N)$-ruling set of $V$. Notice that the each  $v\in I_1$ can check whether it has a neighbor in $I_0$ in $\cd^\star$ model by letting all vertices in $I_0$ transmit, and letting all vertices in $I_1$ listen.
\end{proof}

Our $\broadcast$ algorithm is based on the idea of {\em iterative clustering}, which is similar to the one in Section~\ref{section:diameter}; but the clustering algorithm in this section is based on ruling set. Notice that a ruling set $I$ naturally induces a clustering by letting each vertex in $I$ initiate a cluster, and each remaining vertex joins the cluster of its nearest vertex in $I$.

\paragraph{Clustering by Ruling Set.}  Notice that the proof of Lemma~\ref{lemma:ruling-set} does not generalize to $(k,(k-1)\log N)$-ruling set for $k>2$. The reason is that for a vertex $u \in I_0$ to send a signal to a vertex $v \in N^{k-1}(u) \cap I_1$, an intermediate vertex $w \notin I_0 \cap I_1$ is required.
However, a clustering resulting from a $(2,\log N)$-ruling set $I$ may have size-$1$ clusters.
This issue can be overcome as follows. Observe that any two size-$1$ clusters $\{u\}$ and $\{v\}$ must not be adjacent, since only vertices in $I$ can initiate a cluster. If we add one additional step which merges each size-$1$ cluster into any of their neighboring clusters (which must be of size at least 2), then all clusters will have size at least $2$. Thus, it is guaranteed that the number of clusters int the new clustering is at most half of  that in the old clustering.

\paragraph{Algorithm.} In what follows, we present and analyze our $\broadcast$ algorithm. Some implementation details are deferred to Section~\ref{subsect:cluster-det}, where we describe the cluster structure for deterministic $\cd$, and show how to perform the three operations {\sf Down-cast}, {\sf All-cast} and {\sf Up-cast}. This allows us to simulate one round of a $\cd^\star$ algorithm (that uses messages in $\{1, \ldots, M\}$) on the cluster graph using $O(n \min\{M,N\} N)$ time and $O(\min\{\log M, \log N\} \log N)$ energy.
Thus, the simulation of the ruling set algorithm of Lemma~\ref{lemma:ruling-set} costs $O(n N^2)$ time and $O(\log^2 N)$ energy (since $M=1$).

Suppose that we have computed a $(2,\log N)$-ruling set $I$ of the cluster graph.
Then we can obtain a new clustering with only $|I|$ clusters as follows. Let each cluster $C \in I$ initiate a new cluster, and let each remaining cluster $C'$ joins the new cluster of some $C \in I$. This can be done using $O(\log N)$ iterations of merging clusters (since $I$ is a $(2,\log N)$-ruling set). Each merging operation can be done by applying using {\sf Up-cast} and {\sf Down-cast} (as described in Section~\ref{subsect:maintain}), and it requires transmitting messages of length $O(\log N)$. Each merging operation costs $O(n N^2)$ time and $O(\log^2 N)$ energy (since $M \geq N$). Thus, it takes $O(n N^2 \log N)$ time and $O(\log^3 N)$ energy to compute the new clustering.

After $O(\log n)$ iterations of clustering based on ruling set computation, the whole graph is a single cluster, and thus we can run the algorithm of Lemma~\ref{lem:reduce} to solve $\broadcast$.
We have the following theorem.

\begin{theorem}
There is an $O(nN^2 \log N \log n)$-time and $O(\log^3 N \log n)$-energy deterministic algorithm that solves $\broadcast$ in $\cd$ model.
\end{theorem}

\subsection{Cluster Structure}\label{subsect:cluster-det}
We describe the cluster structure for deterministic algorithms. Each cluster $C$ is a rooted tree, where the cluster center $r$ is the unique layer-0 vertex in $C$, and each vertex $v$ in $C$ is equipped with a unique identifier $\ID(v)$ and a good labeling $\mathcal{L}(v)$. For each $i>0$, we assume that each layer-$i$ vertex $u$ has a designated layer-$(i-1)$ parent $v \in N(u)$, and $u$ knows $\ID(v)$.
Consider the following two tasks (which are slightly different than the ones defined in Section~\ref{subsect:cluster}).
\begin{itemize}
\item {\sf Upward transmission.} Let $i > 0$, and let $V'$ be a subset of layer-$i$ vertices that have some messages to send. The goal is to have each layer-$(i-1)$ vertex $v$ that has a child in $V'$ to receive some message from one of its child.
\item {\sf Downward transmission.} Let $i \geq 0$, and let $V'$ be a subset of layer-$i$ vertices that have some messages to send. The goal is to have each $v \in V'$ deliver its message to all its children.
\end{itemize}

\begin{lemma}\label{lemma:cluster-transmit-det}
In deterministic $\cd$ model, both {\sf Downward transmission} and  {\sf Upward transmission} can be solved using $O(\min\{M,N\} N)$ time and $O(\min\{\log M, \log N\} \log N)$ energy.
\end{lemma}
\begin{proof}
We only present the proof for {\sf Downward transmission}, since {\sf Upward transmission} can be solved analogously.
The algorithm is as follows.
We allocate $N$ time intervals. The $j$th interval is reserved for $\sr$ between the layer-$i$ vertex $v$ with $\ID(v)=j$ and its children (i.e., $S = \{v\}$ and $R$ is the set of all children of $v$).
\end{proof}

The three operations {\sf Down-cast}, {\sf All-cast} and {\sf Up-cast} (they are defined in Section~\ref{subsect:beep}, but the subroutines {\sf Downward transmission} and  {\sf Upward transmission} are the ones defined in this section) can also be done  in our cluster structure.
\begin{itemize}
\item {\sf Down-cast} can be implemented by $n-1$ {\sf Downward transmission} (from $i=0$ to $i = n-2$), and it costs  $O(n \min\{M,N\} N)$ time and $O(\min\{\log M, \log N\} \log N)$ energy.
\item {\sf All-cast} can be implemented by running $\sr$ for $N$ times in a way similar to the proof of Lemma~\ref{lemma:cluster-transmit-det}, and it costs  $O(\min\{M,N\} N)$ time and $O(\min\{\log M, \log N\} \log N)$ energy.
\item {\sf Up-cast} can be implemented by $n-1$ {\sf Upward transmission} (from $i=n-1$ to $i=1$), and it costs  $O(n \min\{M,N\} N)$ time and $O(\min\{\log M, \log N\} \log N)$ energy.
\end{itemize}
Observe that one round of $\cd^\star$ model on cluster graph can be simulated using the three operations {\sf Down-cast}, {\sf All-cast} and {\sf Up-cast}, and so we conclude the following lemma.

\begin{lemma}\label{lemma:simul-cd}
In $\cd$ model, we can deterministically simulate any deterministic algorithm in $\cd^\star$ model on a cluster graph, where each message is an integer in $\{1, \ldots, M\}$. Each round of the algorithm is simulated with $O(n \min\{M,N\} N)$ time and $O(\min\{\log M, \log N\} \log N)$ energy.
\end{lemma}

\end{document}